\documentclass[journal,twoside,web]{ieeecolor}
\usepackage{tmi}
\usepackage{cite}
\usepackage{amsmath,amssymb,amsfonts}
\usepackage{algorithmic}
\usepackage{graphicx}
\usepackage{textcomp}
\usepackage{thmtools, thm-restate}

\usepackage{epsfig}
\usepackage{tabularx}
\usepackage{booktabs}
\usepackage{verbatim}
\usepackage{xcolor}
\usepackage{amstext}
\usepackage{upgreek}
\usepackage{multirow}
\usepackage{bbm}
\usepackage[algo2e,ruled,vlined]{algorithm2e}

\newtheorem{thm}{Theorem}[section]

\usepackage{todonotes}

\newcommand{\mrbrains}[1]{\includegraphics[width=0.19\linewidth]{Figures/MRBrains/#1}}
\newcommand{\mrbrainsretrain}[1]{\includegraphics[width=0.19\linewidth]{Figures/mrb_retrain/#1}}
\newcommand{\ppmi}[1]{\includegraphics[width=0.19\linewidth]{Figures/ImpactOfLambda/#1}}
\newcommand{\robust}[1]{\includegraphics[width=0.24\linewidth]{Figures/Robustness/#1}}

\newcommand{\update}[1]{#1}

\newcommand{\revision}[1]{\color{black}{#1~}\color{black}}

\newcommand{\beq}{\begin{equation}}
\newcommand{\eeq}{\end{equation}}

\renewcommand{\vec}[1]{\mathbf{#1}}
\newcommand{\mr}[1]{\mathrm{#1}}

\newcommand{\expect}{\mathbb{E}}
\newcommand{\real}{\mathbb{R}}
\newcommand{\ttimes}{\,$\times$\,}

\newcommand{\loss}{\mathcal{L}}
\newcommand{\lossDis}{\ell_{D}}
\newcommand{\lossSeg}{\ell_{S}}

\newcommand{\sid}{\mr{id}}
\newcommand{\lossEst}{\widetilde{\loss}}

\newcommand{\batch}{\mathcal{B}}

\newcommand{\dataImg}{\mathcal{X}}
\newcommand{\dataGT}{\mathcal{Y}}
\newcommand{\img}{\vec{x}}

\newcommand{\gt}{\vec{y}}
\newcommand{\same}{s}
\newcommand{\pred}{\widehat{s}}
\newcommand{\gtPix}{y}
\newcommand{\seg}{\widehat{\vec{y}}}
\newcommand{\segPix}{\widehat{y}}
\newcommand{\params}{\uptheta}
\newcommand{\one}{\mathbbm{1}}
\newcommand{\vone}{\mathbbm{1}}

\usepackage[pagebackref=true,breaklinks=true,letterpaper=true,colorlinks,bookmarks=false]{hyperref}

\def\BibTeX{{\rm B\kern-.05em{\sc i\kern-.025em b}\kern-.08em
    T\kern-.1667em\lower.7ex\hbox{E}\kern-.125emX}}
    
\begin{document}
\title{Privacy-Net: An Adversarial Approach for Identity-Obfuscated Segmentation of Medical Images}
\author{Bach Ngoc Kim, 
        Jose Dolz, 
        Pierre-Marc Jodoin, 
        and Christian Desrosiers 
\thanks{This work is supported in part by the Natural Sciences and Engineering Research Council of Canada (NSERC) under Grant RGPIN-2017-758170 and RGPIN-2018-05715.}
\thanks{This work is supported in part by the Reseau de BioImagrie du Quebec (RBIQ) under Grant 247388.}
\thanks{This work is supported in part by NVIDIA corporation through their GPU grant program.}
\thanks{Bach Ngoc Kim is with \'{E}cole de Technologie Sup\'{e}rieure, Montreal, QC H3C 1K3, Canada (e-mail: bachknk49@gmail.com).}
\thanks{Jose Dolz is with \'{E}cole de Technologie Sup\'{e}rieure, Montreal, QC H3C 1K3, Canada (e-mail: jose.dolz@etsmtl.ca).}
\thanks{Pierre-Marc Jodoin is with Universit\'{e} de Sherbrooke, Sherbrooke, QC J1K 2R1, Canada (e-mail: pierre-marc.jodoin@usherbrooke.ca).}
\thanks{Christian Desrosiers is with \'{E}cole de Technologie Sup\'{e}rieure, Montreal, QC H3C 1K3, Canada (e-mail: christian.desrosiers@etsmtl.ca).}}

\maketitle

\begin{abstract}
This paper presents a client/server privacy-preserving network in the context of multicentric medical image analysis. Our approach is based on adversarial learning which encodes images to obfuscate the patient identity while preserving enough information for a target task. Our novel architecture is composed of three  components: 1) an encoder network which removes identity-specific features from input medical images, 2) a  discriminator network that  attempts  to  identify the subject from the encoded images, 3) a medical image analysis network which analyzes the content of the encoded images (segmentation in our case).  By simultaneously fooling the discriminator and optimizing the medical analysis network, the encoder learns to remove privacy-specific features while keeping those essentials for the target task. Our approach is illustrated on the problem of segmenting brain MRI from the large-scale Parkinson Progression Marker Initiative (PPMI) dataset. Using longitudinal data from PPMI, we show that the discriminator learns to heavily distort input images while allowing for highly accurate segmentation results. \update{Our results also demonstrate that an encoder trained on the PPMI dataset can be used for segmenting other datasets, without the need for retraining.} The code will be made available upon acceptance of the paper. 
\end{abstract}


\section{Introduction}
%
%
%
%
\IEEEPARstart{M}{achine} learning models like deep convolutional neural networks (CNNs) have achieved outstanding performances in complex medical imaging tasks such as segmentation, registration, and disease detection~\cite{Zhou2017,litjens2017survey}. However, privacy restrictions on medical data including images impede the  development of centralized cloud-based image analysis systems, a solution that has its share of benefits: no on-site specialized hardware, immediate trouble shooting or easy software and hardware updates, among others.  

While server-to-client encryption can prevent attacks from outside the system, it cannot prevent cybercriminals within the system from gaining access to private medical data.  Another approach to obfuscate the identity of a patient is to anonymize its data. In case of images, this is done by removing the patient-related DICOM tags or by converting it into a tag-free format such as PNG or NIFTI. However, as shown by Kumar et al.~\cite{Kumar2018} and further illustrated in this paper, the raw content of an image can be easily used to recover the identity of a person with up to $97\%$ of accuracy. 

A recent solution for decentralized training on multi-centric data is federated learning~\cite{McMahan2017}. 
The idea behind this strategy is to transfer the training gradients of the data instead of the data itself. While such approach is appealing to train a neural network with data hosted in different hospitals, it does not allow the use of a centralized cloud-based model for making predictions at test time without transmitting patient data.

Another solution for privacy protection is homomorphic-encryption (HE)~\cite{Dowlin16,Hesamifard17,nandakumar2019towards}. Although it ensures absolute data protection, one can also train a neural network with both encrypted and non-encrypted data.  Unfortunately, since the HE operations are limited to multiplication and addition, the non-linear operations of a CNN have to be approximated by polynomial functions which makes neural networks prohibitively slow. For example, \cite{nandakumar2019towards} reports computation times above 30 minutes to process a single 28\ttimes28 image using an optimized network with only 954 nodes. Thus, homomorphic neural networks so far proposed have been relatively simplistic~\cite{Hardy2017} and it is not clear how state-of-the-art medical image analysis CNNs like U-Net~\cite{DBLP:journals/corr/RonnebergerFB15} could be implemented in such framework. Furthermore, HE imposes important communication overhead~\cite{rouhani2018} and its use within a distributed learning framework is still cumbersome~\cite{Hardy2017}.

 \begin{figure*}[tp]
    \centering
    \setlength{\tabcolsep}{3pt}
    \begin{tabular}{c}
    \includegraphics[width=.87\textwidth]{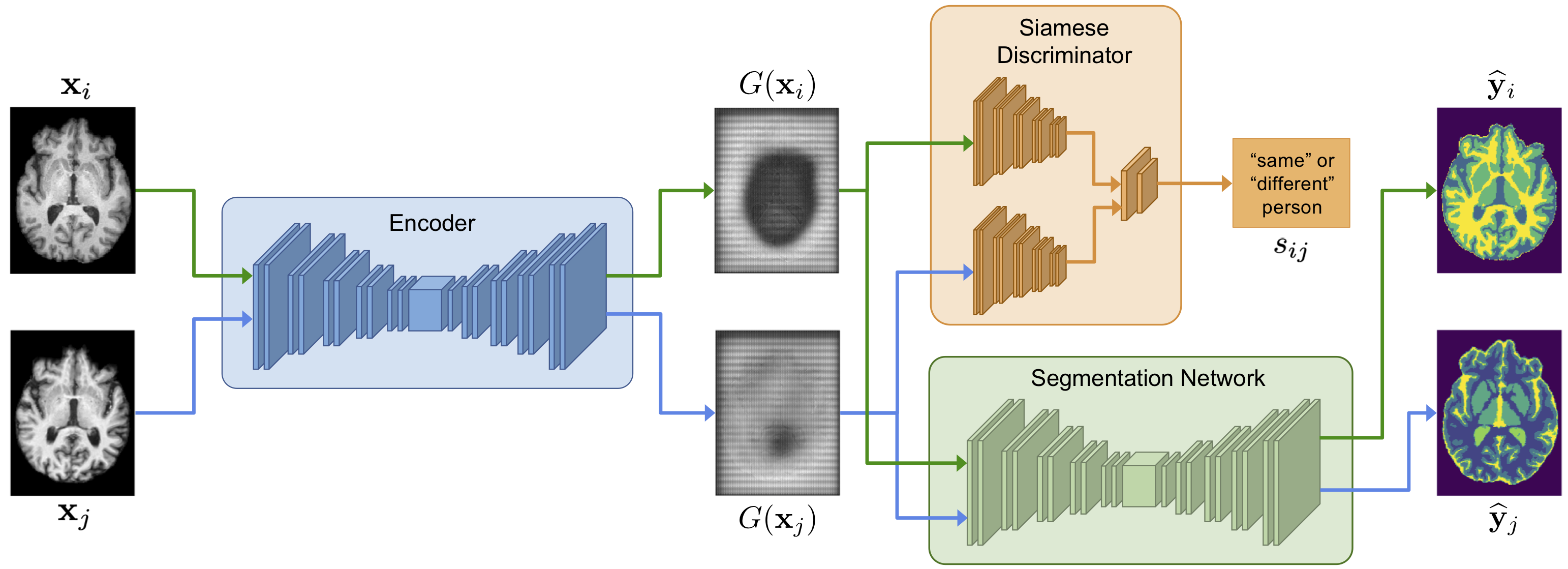} \\
    \update{a) Training configuration} \\[5mm]
    \includegraphics[width=.87\textwidth]{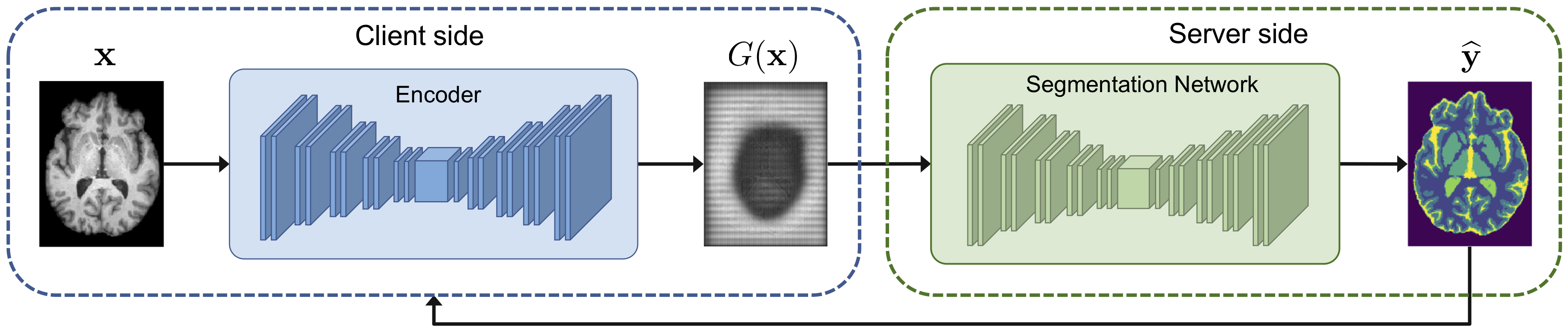} \\    
    \update{b) Test (deployed) configuration}
    \end{tabular}
    \vspace{.1cm}
    \caption{\update{(a)} Training configuration of our proposed system: 1) a client-side encoder network $G$ converts input images $\img_i$ and $\img_j$ into two feature maps $G(\img_i)$ and $G(\img_j)$, 2) the discriminator network $D$ tries to determine if its input data comes from the same patient ($s_{ij}=1$) or not ($s_{ij}=0$), and 3) a server-side segmentation network $S$ segments the encoded images. \update{(b)} At test time, the discriminator is removed from the system and images are processed one at a time by the encoder \update{on the client side} and the segmentation network \update{deployed on the server. The segmentation result is sent back to the client}.
    }
    \label{fig:system}
\end{figure*}

In this paper, we propose a client-server system which allows for the analysis of multi-centric medical images while preserving patient identity. A high-level view of the proposed system is given in Fig.~\ref{fig:system}. On the client side is an encoder that converts patient-specific data into an identity-obfuscated signal containing enough semantic information to analyse its content. The encoded data is then sent to the server where it is analyzed and the results of this analysis are sent back to the client. Since each hospital has the same encoder, the server can keep on updating its system without having access to patient-specific information.

We achieve this with an adversarial learning approach inspired by generative adversarial networks (GAN)~\cite{DBLP:journals/corr/LucCCV16,goodfellow2014generative,ganin2016domain} but with two main differences. As illustrated in Fig.~\ref{fig:system}, instead of being a two-network configuration, our system involves three networks: 1) an image encoder, 2) a discriminator and 3) a medical image analysis network (a segmentation CNN in our case). Whereas the encoder's objective is to obfuscate the content of a raw input image, the goal of the discriminator is to determine whether two encoded images come from the same patient or not. The third network is a CNN which analyzes the content of the encoded image. As such, while the encoder tries to fool the discriminator, it must preserve enough information to allow the third network to successfully analyze its content. At test time, the encoder network residing on the client side converts a raw image $\img$ into an encoded (and yet secure) feature map $G(\img)$. Thereafter, $G(\img)$ is transferred to the cloud-based server where the segmentation network is deployed. The resulting segmentation map $\seg$ is then sent back to the client. 

The major contributions of this work are as follows:
\begin{itemize}
\item We present the first client-server system for semantic medical image segmentation which allows for identity-preserving distributed learning. Obfuscating identity while preserving task-specific information is particularly challenging for segmentation, which requires to assign a label for each image pixel.
\item Our model proposes a novel architecture combining two CNNs, for the encoder and segmentation network, with a Siamese CNN for the discriminator. This Siamese discriminator learns identity-discriminative features from image pairs instead of a single image, allowing us to have a variable number of classes (i.e., subject IDs). Unlike the work in \cite{oleszkiewicz2018siamese}, our model is trained using both an adversarial Siamese loss and a task-specific loss, thereby providing encoded images that obfuscate identity while preserving the information required for the target task. 
\item \update{We provide a theoretical analysis showing that the proposed model minimizes the mutual information between pairs of encoded images and a variable indicating if these images are from the same subject. This analysis motivates our approach from a information theoretic perspective.}
\item We demonstrate 
that the privacy-preserving encoder learned with a given dataset can be used to encode images from another dataset, and that these encoded images are useful to update the segmentation network. 


\end{itemize}

\section{Related Works}

\paragraph{Privacy preserving in visual tasks} Traditional methods to preserve privacy rely on cryptographic approaches \cite{ziad2016cryptoimg,wang2017encrypted} which create local homomorphic encryptions of visual data. Although these methods perform well in some applications, homomorphic cryptosystems typically incur  high computational costs \cite{paillier1999public} and are mostly restricted to simple linear classifiers. This limits their usability in scenarios requiring more complex models like deep neural networks. Another solution consists in extracting feature descriptors from raw images, which are then transferred to the encrypted dataset server \cite{hsu2011homomorphic}. Nevertheless, sensitive information from original images can be still recovered from standard features, making these systems vulnerable to cyberattacks. An alternative strategy is to employ low-resolution images \cite{dai2015towards,chen2016estimating} or image filtering techniques \cite{butler2015privacy,jalal2012depth} to degrade sensitive information. However, since these approaches also reduce the quality of the visual content, they are limited to a reduced set of tasks such as action or face expression recognition. 
More recently, M{c}Clure et al. \cite{mcclure2018distributed} proposed using continual learning to circumvent the issue of privacy preservation in the context of multi-center brain tumor segmentation. Nevertheless, unlike our method, their approach is not directly optimized to obfuscate identity from visual data.

\paragraph{Federated learning} Federate learning has recently emerged as a solution to build machine learning models based on distributed data sets while preventing data leakage \cite{DBLP:journals/corr/XieBFGLN14,DBLP:journals/corr/KonecnyMRR16,DBLP:journals/corr/McMahanMRA16,DBLP:journals/corr/abs-1812-03288,yang2019federated}. With this approach, the learning process involves collaboration from all the data owners without exposing their data to others. This can typically be achieved by sharing the architecture and parameters between the client and server during training, along with intermediate representations of the model that may include the gradients, activations and weight updates. Thus, the client downloads the model from the server and updates the weights based on its local data. Yet, a drawback of these strategies is their huge requirements for network bandwidth, memory and computational power, which strongly limits their scalability. More importantly, federated learning does not prevent, at test time, from having to send private data from the client to the server in a scenario such as ours where the server holds the model and processes the data.  Also, while HE can be combined to federated learning, its communication protocol is cumbersome and imposes important communication overhead~\cite{Hardy2017,rouhani2018}.


\paragraph{Privacy preserving with adversarial learning} The recent success of adversarial learning has led to the increased adoption of this technique for the protection of sensitive information, particularly in visual data. Xu et \textit{al.} \cite{xu2019ganobfuscator} proposed to add carefully-designed noise to gradients during the learning procedure to train a differentially-private GAN in the context of image recognition. An unsupervised utility loss is employed for training in \cite{raval2017protecting}, based on the assumption that removing private characteristics from an image while minimizing changes to the rest of the image yields encoded representations that can be used to learn a target task. However, since the encoding is performed independently of the task, it is potentially sub-optimal for this task. Other works \cite{pittaluga2019learning,wu2018towards,yang2018learning,DBLP:journals/corr/abs-1904-05514} have leveraged adversarial training to jointly optimize privacy and utility objectives. In these works, the mapping functions for the adversarial and task-specific terms are standard classification models where the number of classes is fixed. In \cite{chen2018vgan}, a model which integrates a Variational Autoencoder (VAE) and a GAN is proposed to create an identity-invariant representation of face images. To explicitly control the features to be preserved, they include a discriminator which must predict the identity of the subject in a generated image. As the number of possible labels corresponds to the number of subjects to identify, this approach is not suitable for large-scale applications as the one considered in our work. 

To alleviate the problem of a non-fixed number of classes, \cite{oleszkiewicz2018siamese} uses a Siamese architecture for the discriminator which predicts whether two encoded images come from the same subject or not. \update{This paper focuses on biometrical data (e.g., fingerprint), which are dissimilar in nature from the medical images used in our method, and seeks a very different goal: finding the smallest possible transformation to an image which removes identity information and such that images can later be used by non-specific applications. In contrast, we obfuscate images with the strongest possible transformation so that subject identity cannot be recovered while at the same time the encoded image can be used to train an image analysis (i.e., segmentation) task. 
This translates into important methodological differences. First, while the model in \cite{oleszkiewicz2018siamese}  has a generator and a discriminator, our architecture is composed of three separate networks, i.e., an encoder, a Siamese discriminator and a segmentation network. Second, the final objective is different since we aim at maximizing the same-subject classification error, \emph{as well as} optimizing a task-specific loss related to segmentation. In summary, both the structure and objectives between \cite{oleszkiewicz2018siamese} and our work are different.

The work in \cite{wang2018cross} tackles a task opposite to privacy-preserving image analysis, where faces in input images are rejuvenated while preserving, and not removing as in our work, information related to kinship. This is done by minimizing a discriminative sparse metric learning loss encouraging generated images for members of the same family to be nearby in a low-dimensional subspace. \revision{In \cite{Xia_2020_CVPR}, Xia et \textit{al.} also employed adversarial training to develop a novel Generative cross-domain learning method via Structure-Preserving (GSP). The method attempts to transform target data into the source domain in order to take advantage of source supervision.}
}

\section{Methodology}
\label{sec:methods}

\subsection{Proposed system}

As shown in Fig.~\ref{fig:system}, our system implements a zero-sum game involving three separate CNN networks. At the input of our system is a raw image $\img \in \real^{H\times W\times D}$ (in our case a 3D T1 magnetic resonance image (MRI)). During training, images come in pairs $(\img_i, \img_j) \in \dataImg^2$, $i \neq j$. Each image pair is associated to the corresponding ground-truth segmentation maps $(\gt_i,\gt_j)$ and binary target $s_{ij}$ which equals $1$ when $\img_i$ and $\img_j$ come from the same patient and $0$ otherwise. As mentioned in Section~\ref{sec:datasets}, pairs of images from the same patient are not identical as they were acquired during different acquisition sessions, often months apart. 

The first network of our system is an encoder network $G$ parameterized by $\params_G$. The output of the encoder is a feature map $G(\img) \in \real^{H\times W\times D}$ which can be seen as an encoded version of the input image. 
While the encoder could return feature maps of any size, we chose maps with the same size as the input image $\img$ for the following important reasons. First, it allows preserving the information and spatial resolution of the input image. In contrast, using a compressed representation could lead to loss of details. This is why, for example, state-of-art segmentation networks employ skip connections that concatenate detailed features from downsampling layers with low-resolution features from upsampling layers \cite{dolz20183d,DBLP:journals/corr/RonnebergerFB15}. Second, despite the high spatial resolution of encoding $G(\img)$, it is still more compact than convolutional features of standard networks like VGG which have a lower spatial resolution but a larger number of channels (e.g., 14\ttimes14\ttimes512 = 100,352 features at the last convolutional layer of VGG compared to 224\ttimes224\ttimes1 = 50,176 features for our encoding, in the case of 224\ttimes224 images). Third, it enables a fair comparison of segmentation performance with the model using non-encoded images. \update{Last, preserving the same shape as the input image allows processing the encoding image in sub-regions (i.e., 3D patches), which provides additional protection when these sub-regions are sent in a random order to the server for segmentation.} While training the system, the encoder is fed with a pair of images ($\img_i,\img_j)$ and returns two encoded images $G(\img_i)$ and $G(\img_j)$. Here, $\img_i$ and $\img_j$ are processed individually and not concatenated together. 

The second network is the Siamese discriminator network $D$ with parameters $\params_D$, which is fed with a pair of encoded images. The goal of this network is to determine whether the two images come from the same patient or not. By fooling $D$ (i.e., maximizing its loss), the encoder \update{transforms the images and makes it difficult to identify the patient.} 
Last, the third CNN is the segmentation network $S$ with parameters $\params_S$, whose goal is to recover the correct segmentation map $\gt$ given the encoded image $G(\img)$. During training, both $G(\img_i)$ and $G(\img_j)$ are segmented. For this network, we used the widely-adopted U-Net \cite{DBLP:journals/corr/RonnebergerFB15}, which is very effective at segmenting medical images. 

\subsection{Training losses}

As in most adversarial models, our system is trained with two losses that steer the model in opposite directions. In our case, the training procedure involves a segmentation loss and an adversarial discriminator loss:
\begin{align}
    & \!\!\!\min_{\params_G, \,\params_S} \max_{\params_D} \ \loss(\params_G, \params_S, \params_D) 
        \, = \, \expect_{\img,\gt \sim P(\img, \gt)} \big[ \lossSeg\big(S(G(\img)),\gt\big) \big] \nonumber\\[-1mm]
      & \qquad -\lambda \, \expect_{\img_i, \img_j \sim P(\img)} \Big[\lossDis\big(D(G(\img_i), G(\img_j)), s_{ij}\big) \Big]
     \label{eq:total_loss}
\end{align}
where $\lossSeg$ is attached to the segmentation network, $\lossDis$ is attached to the discriminator, and $s_{ij}\!=\!\one_{\sid(\img_i) = \sid(\img_j)}$ is a binary indicator function indicating whether two encoded images come from the same patient or not.

Using $\seg = S(G(\img))$ as shorthand notation for the predicted segmentation map, we employ the generalized Dice loss \cite{DBLP:journals/corr/SudreLVOC17} to train the segmentation network, i.e.
\begin{equation}
    \lossSeg(\seg, \gt) \ = \ 1 \, - \, \frac{2 \sum_{p} \gtPix_p\,\segPix_p}{\sum_{p} \gtPix_p \, + \, \sum_{p}\segPix_p}.
    \label{eq:seg_loss}
\end{equation}
For the adversarial loss, we want the discriminator to differentiate subject identity in pairs of encoded images $G(\img_i)$, $G(\img_j)$. Here, we define discriminator's classification loss $\lossDis$ using binary cross entropy:
\begin{equation}
    \lossDis(\pred,\same) \ = \ -\same\log\pred \ - \ (1-\same)\log(1-\pred).
    \label{eq:xentropy_loss}
\end{equation}

Like most adversarial models, the parameters of our system cannot be updated all at once through a gradient step. Instead, we first update the encoder and segmentation parameters $\params_D,\params_G$ by taking the following gradient descent step: 
\begin{align}\label{eq:update_seg_enc}
(\params_S^{t+1},\params_G^{t+1}) \ 
    \leftarrow \ (\params_S^t, \params_G^t) \, - \, \eta \nabla\lossEst(\params_G^t,\params_S^t).
\end{align}
The gradient is estimated using random batches of image pairs $\batch \subset |\dataImg|\!\times\!|\dataImg|$, as follows:
\begin{align}\label{eq:grad_seg_enc}
&\nabla\lossEst(\params_G,\params_S) \, = \, \frac{1}{|\batch|} \sum_{(i,j) \in \batch} \!\!\!\nabla_{\params_G,\params_S} \Big[\lossSeg\big(\seg_i, \gt_i\big)\nonumber\\[-0.75mm] 
& \ + \ \lossSeg\big(\seg_j, \gt_j\big) \ - \ \lambda\lossDis\big(D(G(\img_i), G(\img_j)), s_{ij}\big)\Big]
\end{align}
We then update the discriminator parameters by taking a gradient ascent step 
\begin{equation}\label{eq:update_discr}
    \params_D^{t+1} \ \leftarrow \ \params_D^t \, + \, \eta \nabla\lossEst(\params_D^t)
\end{equation}
with the batch gradient computed as
\begin{align}\label{eq:grad_discr}
& \nabla\lossEst(\params_D) \, = \, -\frac{\lambda}{|\batch|} \sum_{(i,j) \in \batch} \!\!\! \nabla_{\params_D} \lossDis\big(D(G(\img_i), G(\img_j)), s_{ij}\big).
\end{align}
Details of our training method are provided in Algo. \ref{algo}.

\SetAlFnt{\small\sffamily}
\begin{algorithm2e}[t!]
\SetNoFillComment
%
\KwIn{Images $\dataImg$ and ground-truth masks $\dataGT$} 
\KwOut{Network parameters $\params_G, \params_D, \params_S$}

\BlankLine
\tcc{Initialization}
Initialize network parameters $\params_G, \params_D, \params_S$\;

\BlankLine
\tcc{Main loop}

\For{$\mr{epoch} = 1, \ldots, E_{\mr{max}}$}{

    \For{$\mr{iter} = 1, \ldots, T_{\mr{max}}$}{
        \BlankLine
        Randomly select batch $\batch \subset |\dataImg|\!\times\!|\dataImg|$\;        

        \BlankLine
        Update encoder and segmentation network parameters ($\params_S,\params_G$) using Eq. (\ref{eq:update_seg_enc}) and (\ref{eq:grad_seg_enc})\;  
    
        \BlankLine
        Update discriminator parameters ($\params_D$) using Eq. (\ref{eq:update_discr}) and (\ref{eq:grad_discr})\;
    }

}

\Return{$\params_G, \params_D, \params_S$}\;
\caption{Privacy-preserving network learning}\label{algo}

\end{algorithm2e}

\update{
\subsection{Link to mutual information minimization}
\label{sec:motivation}

The idea of using adversarial learning to obfuscate identity is well-grounded on the principles of information theory. Hence, it can be shown that training a subject-ID classifier as discriminator in an adversarial learning model implicitly minimizes the mutual information between the encoded image  and the corresponding subject ID. However, as mentioned before, this strategy is ill-suited to our problem since the number of classes (i.e. the number of subject ID) is not fixed and instead increases as new subjects are added to the system. This poses two major problems: 1) the output size of $D$ varies over time, and 2) the classification task is hard to learn due to the large number of classes compared to the very low number of samples per classes (i.e., 1--4 images per subject). This motivates our approach, based on a Siamese discriminator, where the identification task is to determine if two encoded images are from the same subject. This approach can naturally incorporate new subjects/classes over time and is easier to learn since it corresponds to a binary classification problem and training samples are more abundant (i.e., image pairs instead of images).

As a theoretical contribution of this work, we show that our proposed privacy-preserving learning approach based on a Siamese discriminator also relates to mutual information minimization. This is done in the following theorem.

\begin{thm}
Let $\img$, $\img'$ be two images, and $G(\img)$, $G(\img')$ be their encoded version obtained by the generator $G$. Denoting as $z\!=\!\sid(\img)$ the subject ID of image $\img$ and $s = \vone_{z=z'}$ the random variable indicating whether images $\img$ and $\img'$ are from the same subject, optimizing the problem defined in Eq. (\ref{eq:total_loss}) corresponds to minimizing mutual information $I(G(\img), G(\img'); \, s)$ between encoded images $G(\img)$, $G(\img')$ and random variable $s$.
\end{thm}
\begin{proof}
\revision{
We proceed by defining mutual information and then bounding it using a variational approach. The mutual information $I(G(\img), G(\img'); \, s)$ between encoded images $G(\img)$, $G(\img')$ and random variable $s$ can be defined as
\begin{align}
 & I\big(G(\img), G(\img'); s\big) \ = \
    H\big(s\big) \, - \, H\big(s \, | \, G(\img), G(\img')\big)\\
    & \ = \ H\big(s\big) \, + \,
    \expect_{\img, \,\img' \sim P(\img, \img')} \Big[\expect_{s' \sim P(s \, | \, G(\img), G(\img'))}\big[\nonumber\\
    & \qquad\qquad\qquad\log P\big(s' \, | \, G(\img), G(\img')\big)\big]\Big]\\
    & \ = \ H\big(s\big) \, + \,
    \expect_{s \sim P(s), \,\img, \,\img' \sim P(\img, \img' | s)} \Big[\expect_{s' \sim P(s \, | \, G(\img), G(\img'))}\big[\nonumber\\
    & \qquad\qquad\qquad\log P\big(s' \, | \, G(\img), G(\img')\big)\big]\Big]
    \label{eq:proof_1}
\end{align}
with $H(x)$ being the Shannon entropy of a random variable $x$ and using the fact that $G$ is a deterministic function. To deal with the intractable computation of $P(s \, | \, G(\img), G(\img'))$, we derive a lower bound using variational distribution $Q(s \, | \, G(\img), G(\img'))$:
\begin{align}
 (\ref{eq:proof_1}) \ = \ & H\big(s\big) \, + \, \expect_{s \sim P(s), \,\img, \,\img' \sim P(\img, \img' | s)}\Big[\expect_{s' \sim P(s | G(\img), G(\img'))}\big[\nonumber\\
    & \underbrace{D_{\mr{KL}}\big(P\big(s' \, | \, G(\img), G(\img')\big) \, \big\| \, Q\big(s' \, | \, G(\img), G(\img')\big)\big)}_{\geq 0} \nonumber\\
    & \quad + \,  \log Q\big(s' \, | \, G(\img), G(\img')\big)\big] \Big]
    \label{eq:proof_3}
\end{align}
\begin{align}
    \ \geq \ & H\big(s\big) \, + \, 
        \expect_{s \sim P(s), \,\img, \,\img' \sim P(\img, \img' | s)}\Big[\expect_{s' \sim P(s | G(\img), G(\img'))}\big[ \nonumber\\
        & \qquad\qquad\log Q\big(s' \, | \, G(\img), G(\img')\big)\big]\Big]
        \label{eq:proof_2}
\end{align}
Next, we use the fact that, for random variables $X$, $Y$ and function $f(x,y)$, $\expect_{x \sim X, \,y \sim Y|x}\big[f(x,y)\big] = \expect_{x \sim X, \,y \sim Y|x, \, x' \sim X|y}\big[f(x',y)\big]$ (see Appendix A.1 of \cite{DBLP:journals/corr/ChenDHSSA16} for proof) to get:
\begin{align}
 (\ref{eq:proof_2}) \ = \ & H\big(s\big) \, + \, 
    \expect_{s \sim P(s), \, \img, \,\img' \sim P(\img, \img' | s)}\big[
    \nonumber\\
        & \qquad \log Q(s \, | \, G(\img), G(\img'))\big]\\
    \ = \ & H\big(s\big) \, + \, 
    \expect_{s \sim P(s), \, z, z' \sim P(z, z' | s), \,
    \img \sim P(\img | z), \, \img' \sim P(\img' | z')}\big[
    \nonumber\\
        & \qquad \log Q(s \, | \, G(\img), G(\img'))\big]\label{eq:lower_bound}
\end{align}
Last, we equate (\ref{eq:lower_bound}) with Eq. (\ref{eq:total_loss}) using the following: 1) $H(s)$ can be treated as a constant, and 2) the variational distribution $Q$ is modeled using our Siamese discriminator $D$, and 3) $\log Q(s \, | \, G(\img), G(\img'))$ is equal to minus the cross-entropy loss $\ell_D$ of Eq. (\ref{eq:xentropy_loss}). Maximizing the lower bound in (\ref{eq:lower_bound}) thus increases its tightness to $I\big(G(\img), G(\img'); \, s\big)$, the two becoming equal when $P(\cdot \, | \, G(\img), G(\img'))=Q(\cdot \, | \, G(\img), G(\img'))$. Consequently, optimizing the loss function of (\ref{eq:total_loss}) minimizes a maximally-tight bound to mutual information.
}
\end{proof}

By minimizing mutual information, which is a symmetric measure of co-dependence between two variables, we ensure that subject identity cannot be established by matching an encoded image with those previously seen in the system. Moreover, a powerful property of mutual information is that it is invariant to any monotone and uniquely invertible transformation of the variables \cite{kraskov2004estimating}. Consequently, it provides a certain robustness to small transformations applied to the (encoded) images, such as translation and rotation. Similarly, it avoids the trivial and non-obfuscating solution where the discriminator is forced to systematically flip its predictions to $s_{\mr{flip}} = 1-s$, since this does not change the mutual information, i.e. $I(G(\img), G(\img'); \, s) = I(G(\img), G(\img'); \, s_{\mr{flip}})$.
}

\subsection{Implementation details}
\label{sec:implementation}

In this study, we used a U-Net architecture ~\cite{DBLP:journals/corr/RonnebergerFB15} but with 3D convolution kernels both for the encoder and the segmentation network. The discriminator is a Siamese network as in \cite{Koch2015SiameseNN}. We used a DenseNet architecture \cite{DBLP:journals/corr/HuangLW16a} with 3D convolution kernels for the CNN backbone. The CNN Siamese backbone (i.e. the left-most CNN inside the discriminator box in Fig.~\ref{fig:system}) is used to extract the features of input images. The last layer of the discriminator contains two fully-connected layers to predict if two encoded images are from the same patient. 

The system was implemented with Pytorch. We used the Adam optimizer with a learning rate of $10^{-4}$ for the whole training process. The PC used for training is an Intel(R) Core(TM) i7-6700K 4.0GHz CPU, equipped with a NVIDIA GeForce GTX 1080Ti GPU with 12 GB of memory. Training our framework takes roughly 30 minutes per epoch, and around 2 days for the fully-trained system. 
 
Since our networks employ 3D convolutions, and due to the large size of MRI volumes, dense training cannot be applied to the whole volume. Instead, volumes are split into smaller patches of size 64\ttimes64\ttimes64, which allows dense training in our hardware setting. \update{During training, the patches are randomly cropped from the MRI volume. In testing, the volume to segment is instead divided in evenly-spaced 3D patches, which are then segmented separately. Individual patch outputs are then combined to obtain full-size segmentation maps. The process of cropping patches and recombining outputs into full-size segmentation maps is detailed in Appendix \ref{sec:appendix2}. 

An important advantage of segmenting patches separately is that they can be sent in a random order to the server once the image has been encoded on the client side. This makes obtaining the identity of a subject even more challenging, since a potential attacker must either reorder patches to recover the full-size encoded image and segmentation map, or match small-size patches with previously seen ones. In Section \ref{sec:retrieval-seg}, we illustrate this advantage in our experiments by performing a subject-ID retrieval analysis on output patches instead of full-size segmentation maps.    
} 

To help the learning process in early training stages, the encoder is pre-trained using an auto-encoder loss. Hence, when the real training starts, the encoder generates encoded images which are almost identical to input ones. Likewise, both the segmentation and discriminator networks were pre-trained on the original images from the the PPMI dataset.

\section{Experimental results}
\label{sec:experiments}

\subsection{Datasets}\label{sec:datasets} 

\paragraph{PPMI} We experiment on brain tissue segmentation of 5 classes: white matter (WM), gray matter (GM), nuclei, internal cerebrospinal fluid (CSF int.) and external cerebrospinal fluid (CSF ext.). We used the T1 images of the publicly-available Parkinson's Progression Marker Initiative (PPMI) dataset~\cite{marek2011ppmi}. We took images from 350 subjects, most of which with a recently diagnosed Parkinson disease. 
Each subject underwent one or two baseline acquisitions and one or two acquisitions 12 months later for a total of 773 images. PPMI MR images were acquired on Siemens Tim Trio and Siemens Verio 3 Tesla machines from 32 different sites. The images have been registered onto a common MNI space and resized to 144\ttimes192\ttimes160 with a 1\,mm$^3$ resolution. More information on the MRI acquisition and processing can be found online: \url{www.ppmi-info.org}.

The dataset was divided into a training and a testing set as shown in Table~\ref{table:PPMI_split}. \update{We split the data in a stratified manner so that images from the same subject are not included in both the training and testing sets.} In order to keep a good balance between the pairs of images, during training and testing, we randomly sampled an equal number of negative and positive samples. Due to the burden of manually annotating volumetric images, we resort to Freesurfer to obtain the segmentation ground-truth, similar to recently-published approaches on large-scale datasets \cite{dolz20183d,roy2019quicknat}.
\update{We use cross-validation to measure the performance of our approach and properly set the hyper-parameters. The training set was randomly divided into 5 stratified subsets, each containing around 54 subjects. We then trained our system for 5 rounds, each time using a different group of 4 subsets for training and the remaining one for validation. 
After validation, we retrained the system on the entire training set and reported results from the independent test set. 
}

\begin{table}[tp]
    \caption{PPMI data used for training and testing our method.
    }
    \label{table:PPMI_split}
    \centering\small
    \begin{tabular}{l|cc|c}
        \toprule
         & Training & Testing & Total\\
        \midrule
        \revision{Num.} of subjects & 269 & 81 & 350 \\ 
        \revision{Num.} of images & 592 & 181 & 773 \\ 
        \revision{Num.} of positive pairs & 509 & 148 & 657\\
        \bottomrule
    \end{tabular}
\end{table}

\begin{figure*}[tp]
    \centering
    \begin{small}
    \setlength{\tabcolsep}{3pt}
    \begin{tabular}{ccccc}
        \ppmi{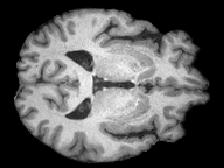} &
        \ppmi{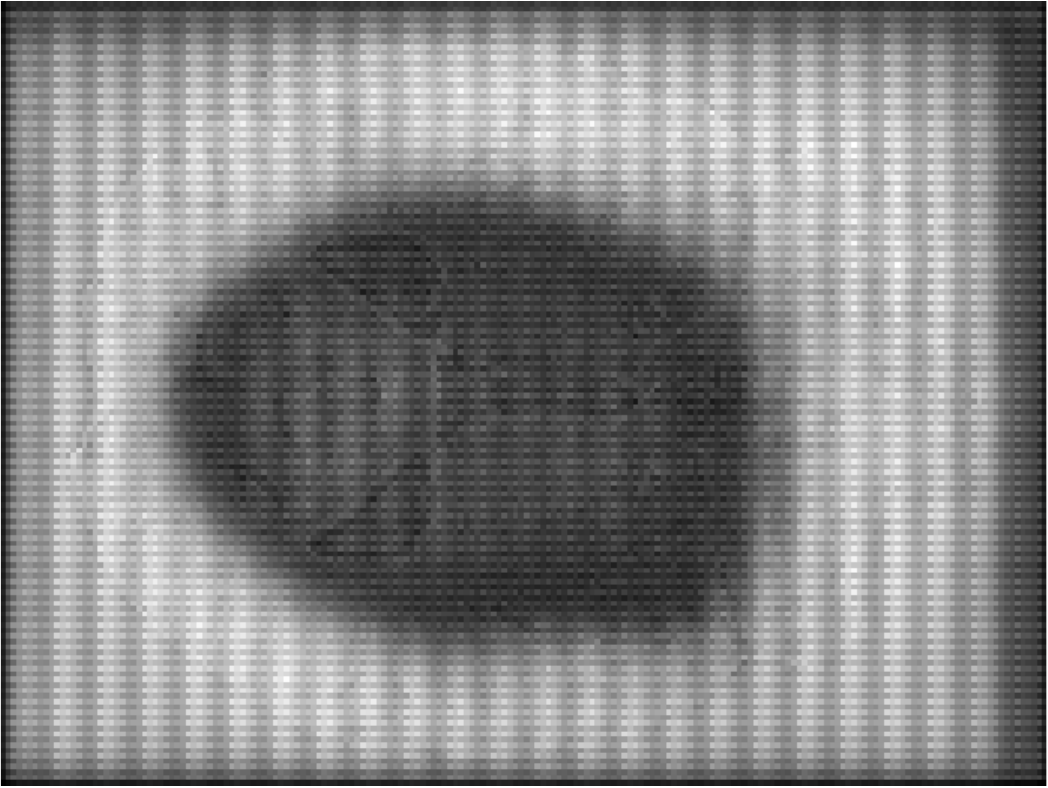} &
        \ppmi{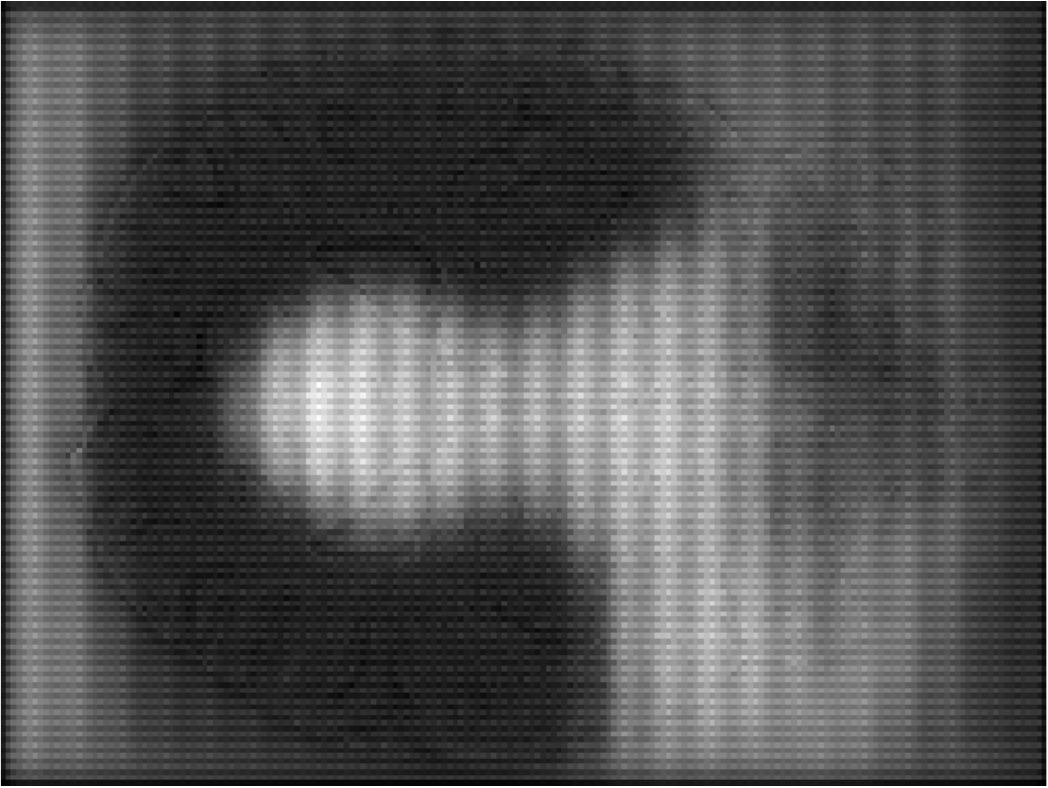} &
        \ppmi{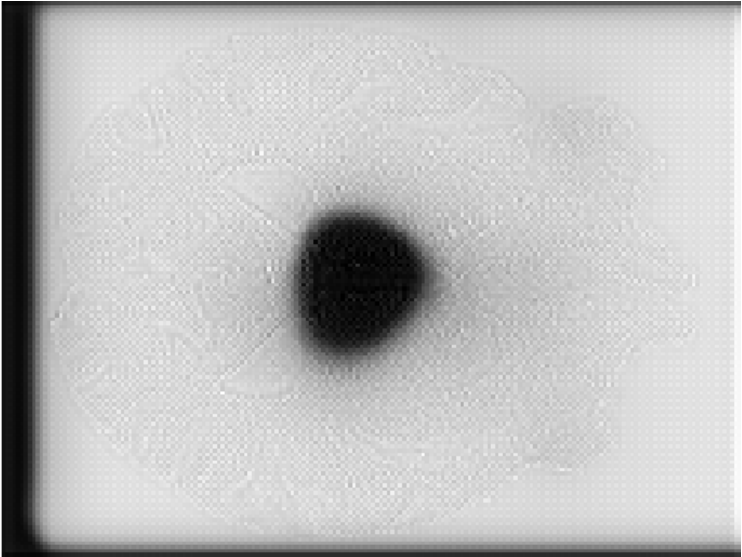} &
        \ppmi{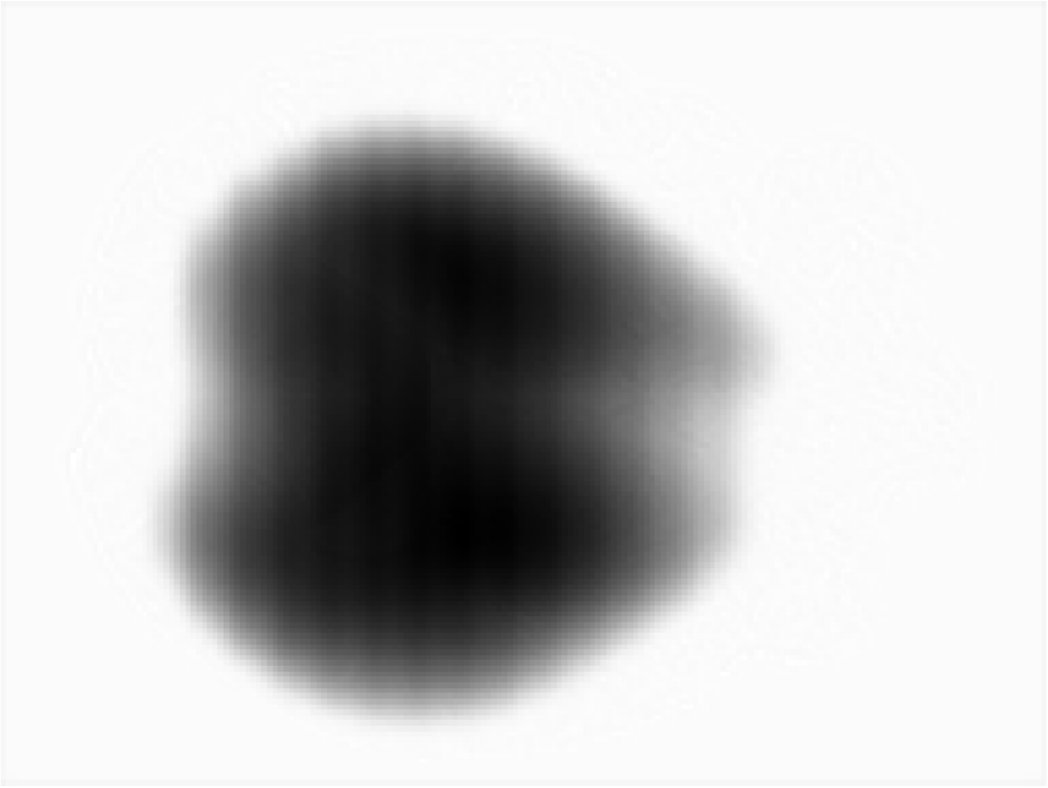} \\
        \ppmi{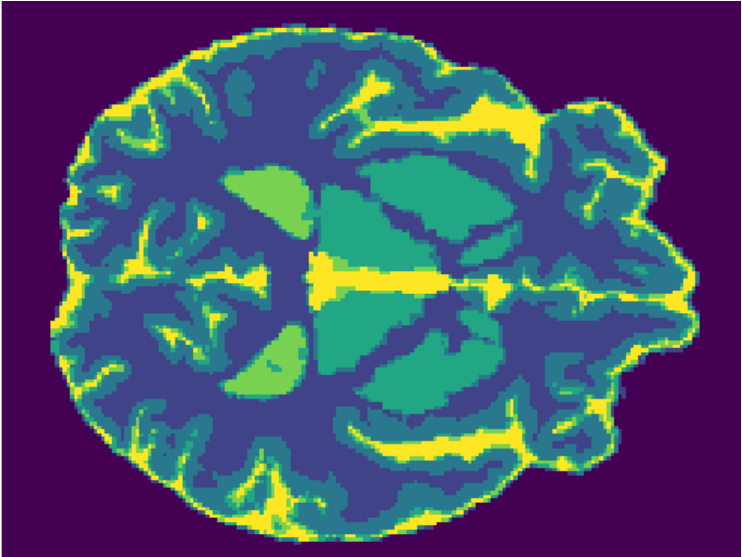} &
        \ppmi{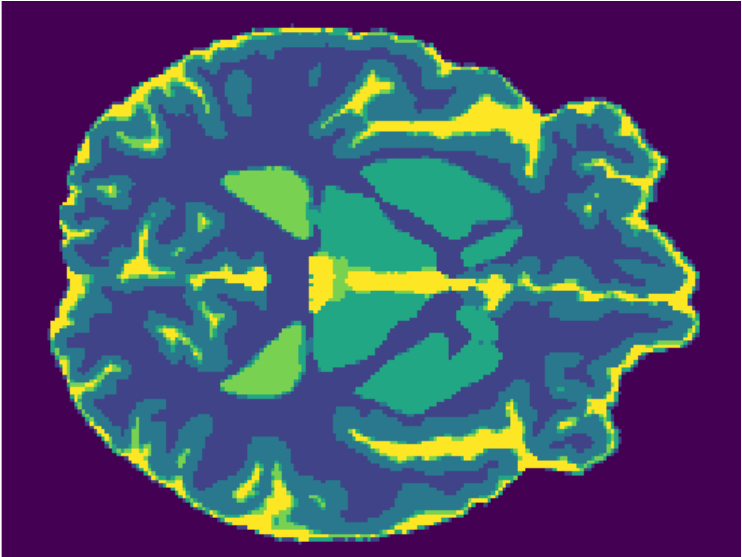} &
        \ppmi{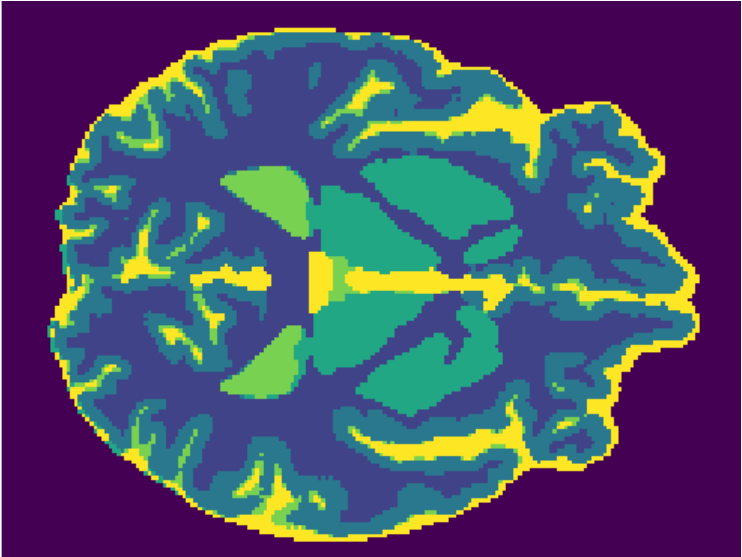} &
        \ppmi{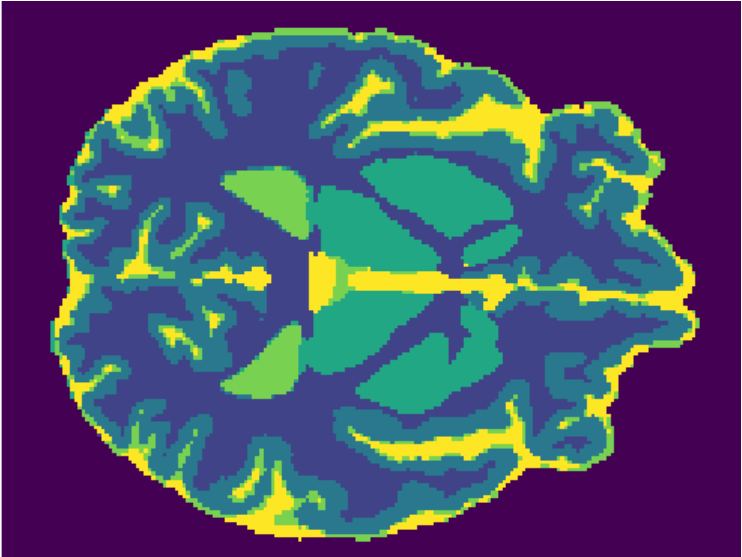} &
        \ppmi{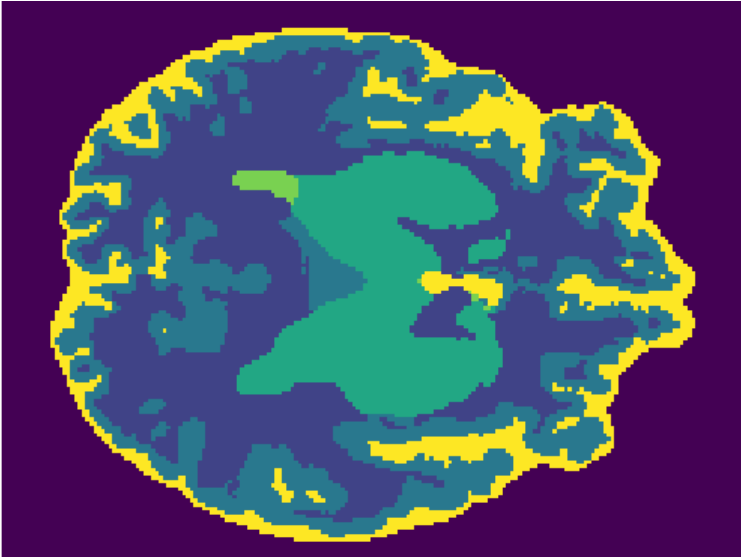} \\[1.25mm]
        \ppmi{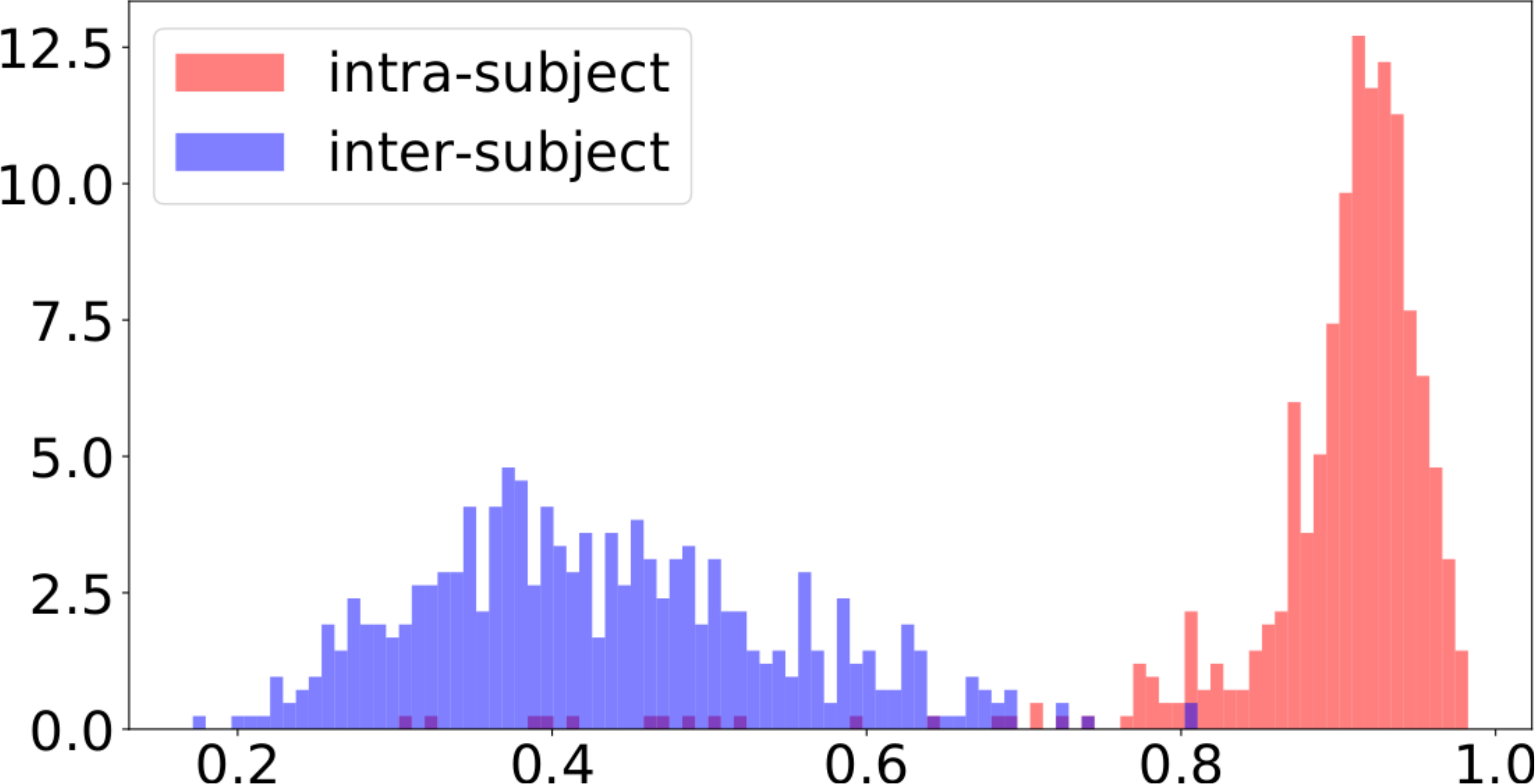} &
        \ppmi{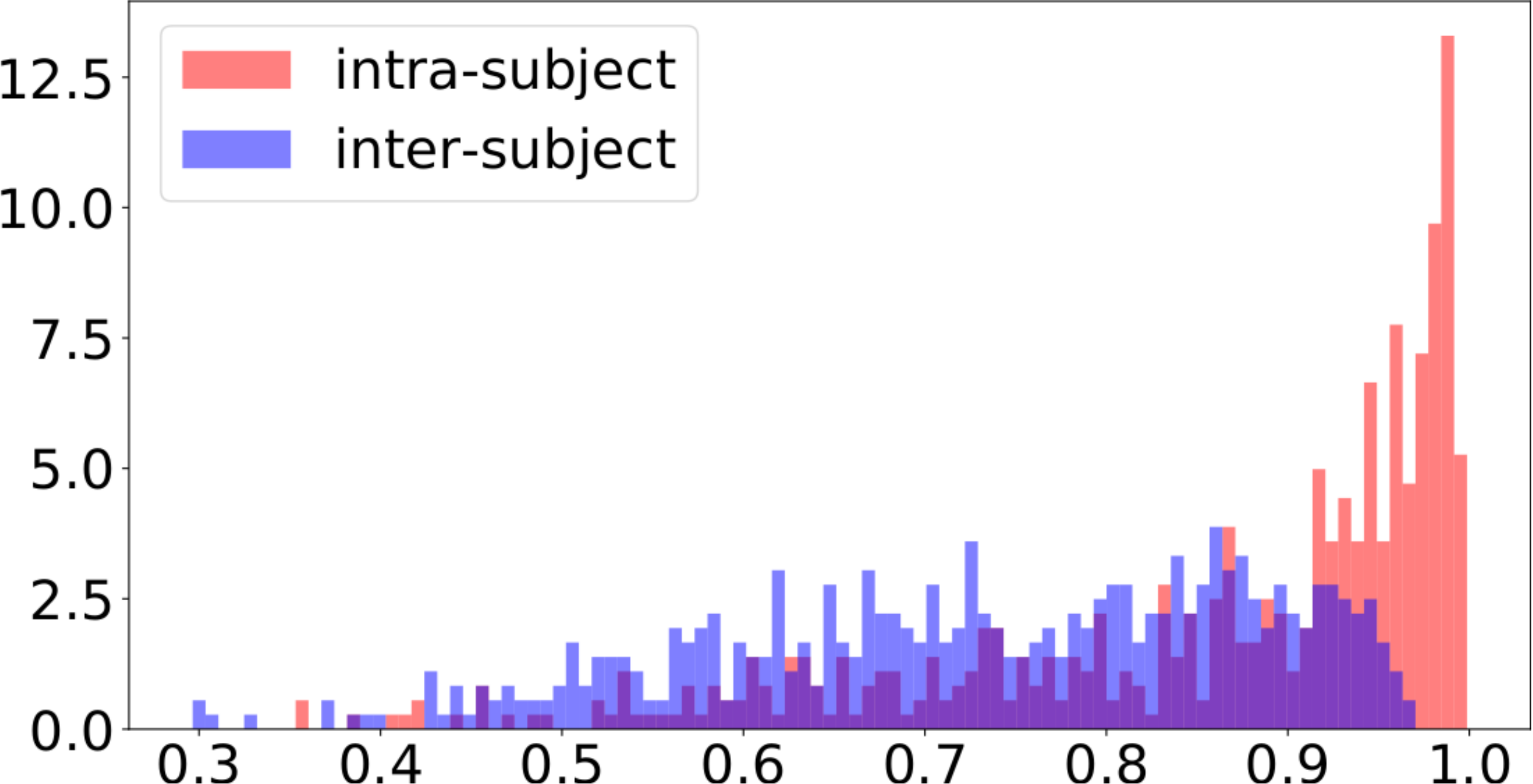} &
        \ppmi{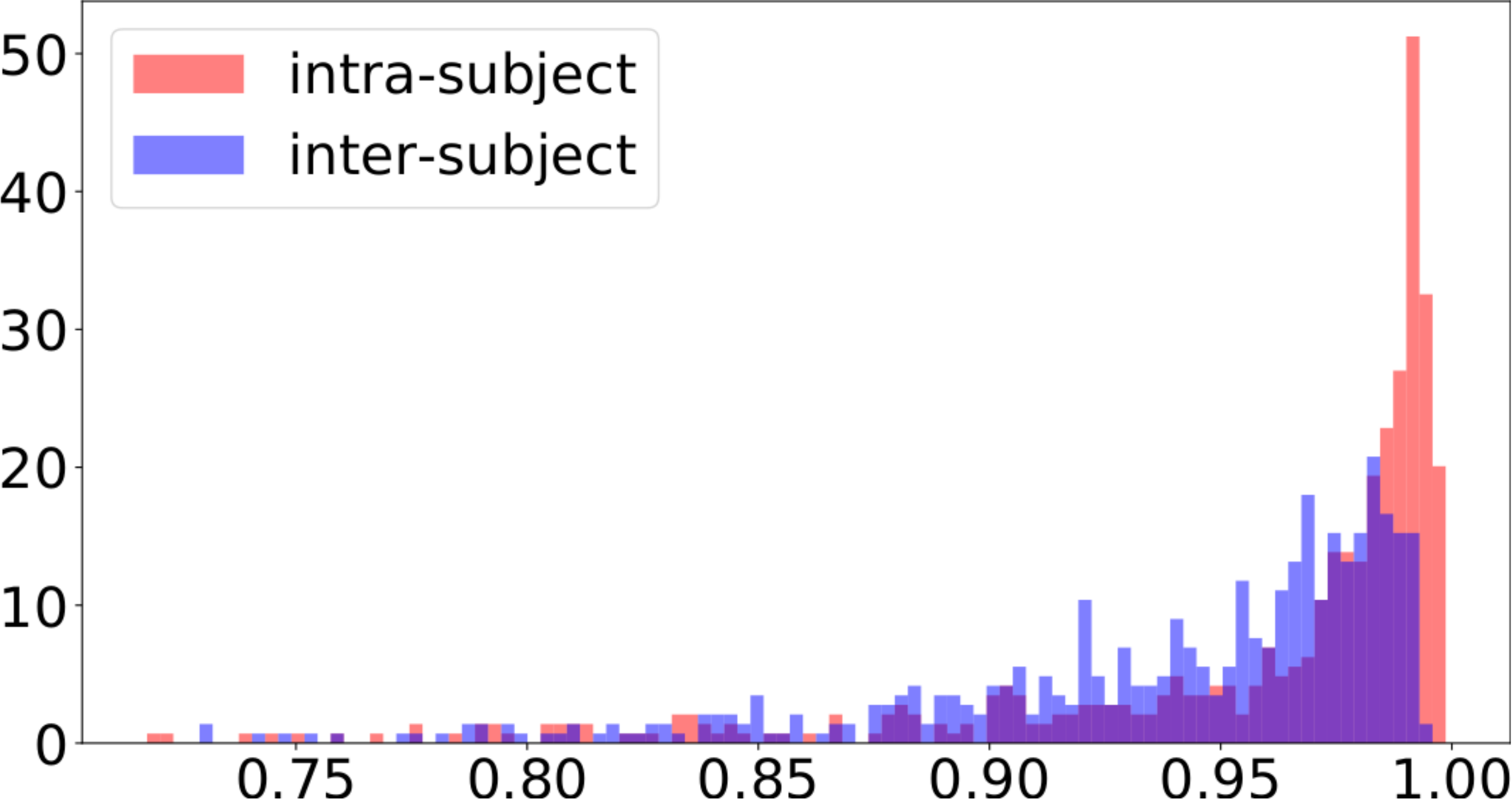} &
        \ppmi{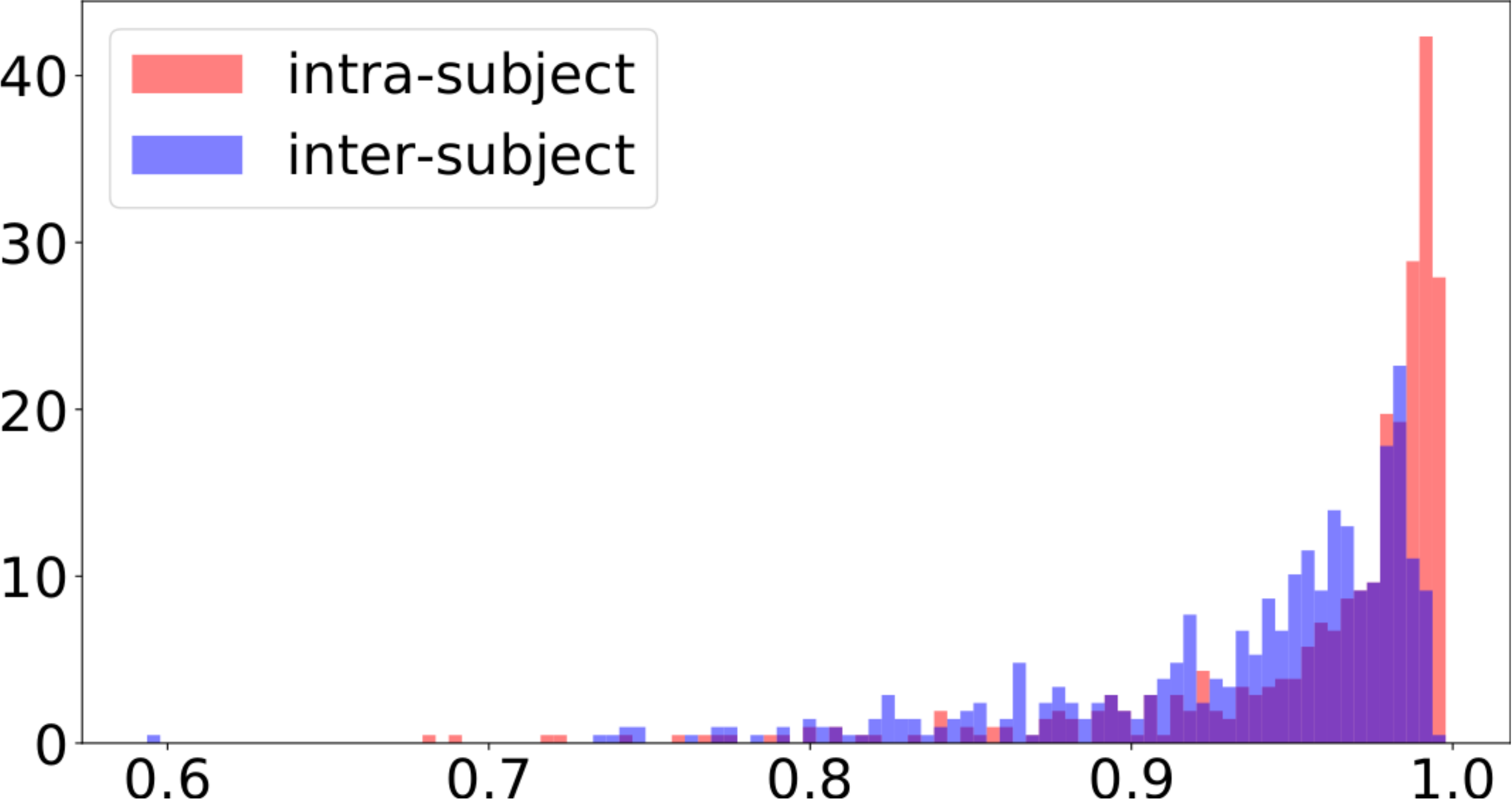} &
        \ppmi{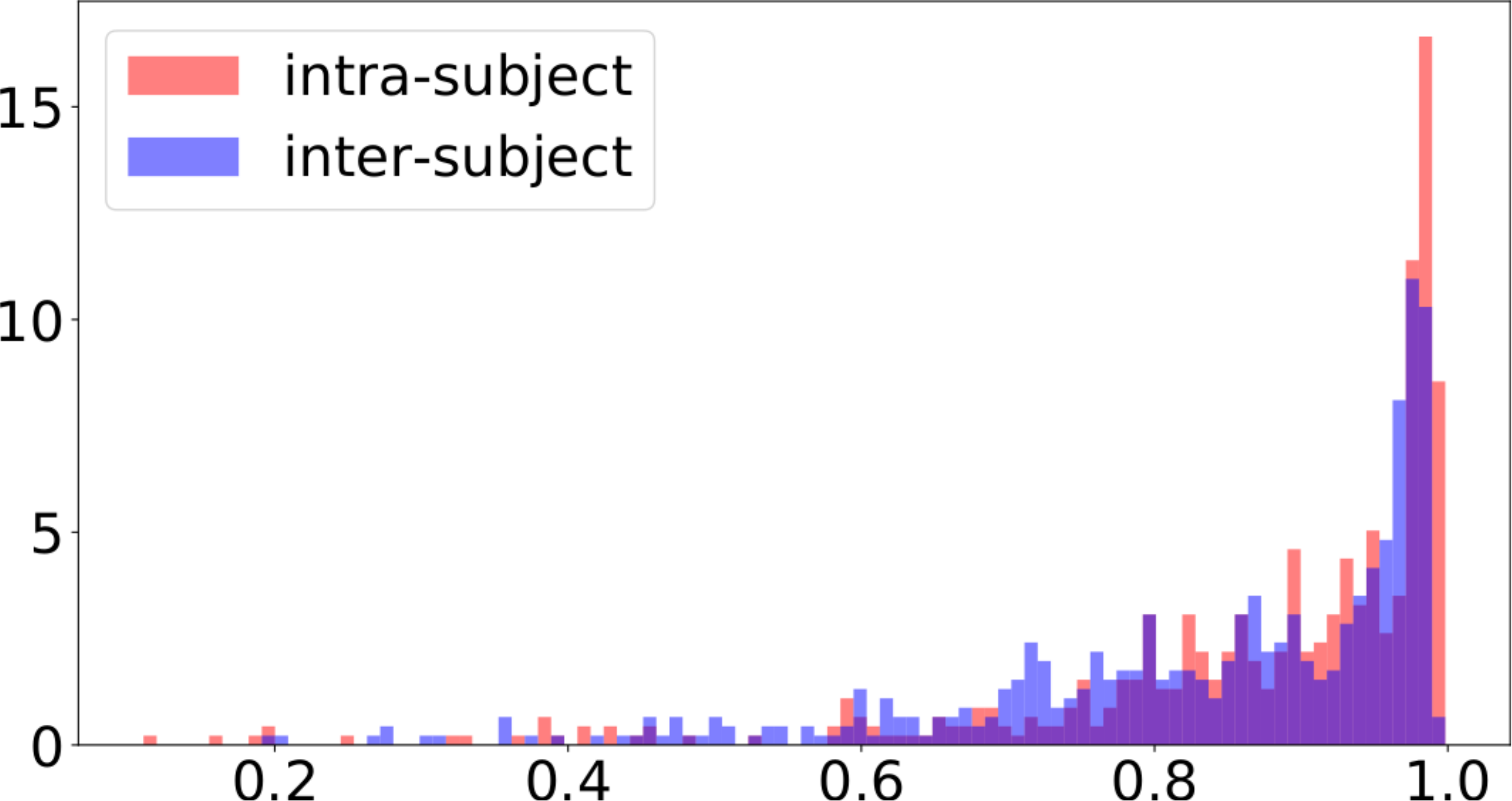} \\
        & $\lambda$\,=\,1 & $\lambda$\,=\,3 & $\lambda$\,=\,10 & $\lambda$\,=\,100
    \end{tabular}
    \vspace{.1cm}
    \end{small}
    \caption{Impact of discriminator loss $(\lambda)$. [\textbf{First column}] (Top row): input MRI image $\img$, (Second row): ground truth segmentation map $\gt$, (Third row): distribution of inter-\update{subject} and intra-subject MS-SSIM score on the PPMI dataset. [\textbf{Remaining columns}], (Top row): encoded image $G(\img)$, (Second row): predicted segmentation $\seg$ and (Third row): distribution of MS-SSIM values between encoded images $G(\img_i)$ and $G(\img_j)$. 
    }
    \label{fig:impact-lambda}
\end{figure*}

\paragraph{MRBrainS} To further validate the proposed method and investigate its generalization ability, we also tested it on segmenting MRI scans from the MRBrainS 2013 challenge dataset \cite{mendrik2015mrbrains}. These images were acquired on a 3.0T Philips Achieva MR scanner 
and come with expert-annotated segmentation masks including three classes: WM, GM and CSF. We employed a single modality (i.e., MR-T1) in our experiments. Bias correction was performed as a pre-processing step. Original images had a resolution of 0.96\ttimes0.96\ttimes3\,mm$^3$ and were registered onto the MNI space using ANTs \cite{Avants2011}.

\subsection{Evaluation metrics}

To gauge the performance of our system, we use the classification accuracy for measuring the discriminator's ability to identify images from the same person, and employ the Dice score for evaluating segmentation results. We also use the multiscale structural-similarity (MS-SSIM) score to measure image-to-image distance as a proxy of perceived image quality \cite{Wang2003}. 

Last, to determine if an encoded image can be used to recover the subject in a top-$k$ retrieval setting \cite{Kumar209726}, we use mean average precision (mAP). \revision{
Given an image $i$, we rank other images in the dataset by their similarity to image $i$. The similarity between two images is the cosine similarity between the feature vectors of each image extracted by the CNN backbone of the Siamese discriminator. Let $T_i$ be the set of images of the same subject as image $i$, and denote as $S^k_i$ the set containing the $k$ images most similar to $i$ (i.e., the $k$ nearest neighbors of $i$). For a given value of $k$, we evaluate the retrieval performance using the measure of top-$k$ precision (also known as precision-at-$k$):
\begin{align}
        \big(\mr{precision}@k\big)_i = \frac{T_i \cap S^k_i}{k}
\end{align}
}

Considering each encoded test image $G(\img_i)$ as a separate retrieval task where one must find other encoded images from the same person, the average precision (AP) for $G(\img_i)$ is given by
\begin{equation}
    \mr{AP}_i \ = \ \frac{1}{\sum_{j\neq i} s_{ij}} \sum_{k=1}^{|\dataImg|} \big(\mr{precision}@k\big)_i \cdot s_{ik},
\end{equation}
where $\mr{precision}@k$ is the precision at cut-off $k$, i.e. the ratio of $k$ encoded images most similar to $G(\img_i)$ which belong to the same person. mAP is then the mean of AP values computed over all test examples.

\subsection{Results}
\label{sec:results}

\subsubsection{Results on non-encoded images}

We first processed the dataset without the adversarial component, i.e, by independently training the segmentation and the discriminator networks without the encoder. We call this setting {\em non-encoded} in our results. In the first row of Table~\ref{table:ClassificationResults}, we see that the discriminator obtains a testing accuracy of 95.3\%. This underlines how easy it is for a neural network to recognize a patient based on the content of a brain MRI. More surprising is the 97\% classification accuracy that we obtain by simply thresholding the image-to-image MS-SSIM score. This can be explained by the inter-\update{subject} and intra-subject MS-SSIM distribution plots shown in the third row of the first column of Fig.~\ref{fig:impact-lambda}. As can be seen, when considering non-encoded images, the intra-subject MS-SSIM scores (red curve) are significantly larger than that of the inter-subjects (blue curve). This again illustrates the ease of recognizing the identity of a person based on the content of a medical image. 

The PPMI segmentation Dice scores on non-encoded images for the five brain regions are in the first row of Table~\ref{table:SegmentationResults}. We also report the overall Dice computed as the mean of Dice scores in all regions, weighted by the regions' size. These results correspond roughly to those obtained in recent publications for the same architecture~\cite{dolz2019}. Note that the nuclei and the internal CSF have a lower Dice due to the smaller size of these regions.

\begin{table}[tp]
    \centering
    \caption{Intra\update{-subject} and inter-subject prediction accuracy on test examples obtained by thresholding MS-SSIM scores, using the adversarial discriminator ($D_{\mr{adv}}$), or training a separate discriminator on the encoded image ($D_{\mr{new}}$). The mAP column is the mean average precision of a top-$k$ retrieval analysis using the Siamese discriminator's embedding as representation. Results are reported for non-encoded images or encoded images for different $\lambda$ values.}
    \label{table:ClassificationResults}    
    \small
    \resizebox{0.47\textwidth}{!}{
    \begin{tabular}{ll|ccc|c}
        \toprule
        & & \multicolumn{3}{c|}{Accuracy} & \multirow[b]{2}{*}{mAP} \\
        \cmidrule(l{6pt}r{6pt}){3-5}
        & & MS-SSIM & $D_{\mr{adv}}$ & $D_{\mr{new}}$ &  \\
        \midrule
        \multicolumn{2}{l|}{Non-encoded} & 0.970 & ~~--- & 0.953   & 0.850 \\
        \midrule        
        \multirow{4}{*}{Encoded}
        & $\lambda$\,=\,1 & 0.564 & 0.520 & 0.598 & 0.189  \\
        & $\lambda$\,=\,3 & 0.533 & 0.537 & 0.615 & 0.152 \\
        & $\lambda$\,=\,10 & 0.510 & 0.523 & 0.577 & 0.141  \\
        & $\lambda$\,=\,100 & 0.503 & 0.516 & 0.513 & 0.087 \\
        \bottomrule
    \end{tabular}
    }
\end{table}

\begin{table}[h]
\caption{Segmentation Dice score on the PPMI test set for different values of $\lambda$. Non-enc refers to the model trained with non-encoded images. 
}
\label{table:SegmentationResults}
\centering
\small\setlength{\tabcolsep}{4pt}
\begin{tabular}{ l|ccccc|c}
    \toprule
    &  GM & WM & Nuclei & CSF int. & CSF ext. & Overall\\
    \midrule
    Non-enc & 0.941 & 0.853 & 0.657 & 0.665 & 0.825& 0.848 \\
    $\lambda$\,=\,1 & 0.925 & 0.824 & 0.580 & 0.598 & 0.752 & 0.812 \\
    $\lambda$\,=\,3 & 0.899 & 0.793 & 0.549 & 0.550 & 0.693 & 0.778 \\
    $\lambda$\,=\,10 & 0.881 & 0.796 & 0.555 & 0.531 & 0.685 & 0.771 \\
    $\lambda$\,=\,100 & 0.847 & 0.692 & 0.454 & 0.405 & 0.513 & 0.684 \\
    \bottomrule
\end{tabular}
\end{table}

\begin{figure*}[tp]
    \centering
    \begin{small}
    \setlength{\tabcolsep}{3pt}
    \begin{tabular}{ccccc}
        \mrbrains{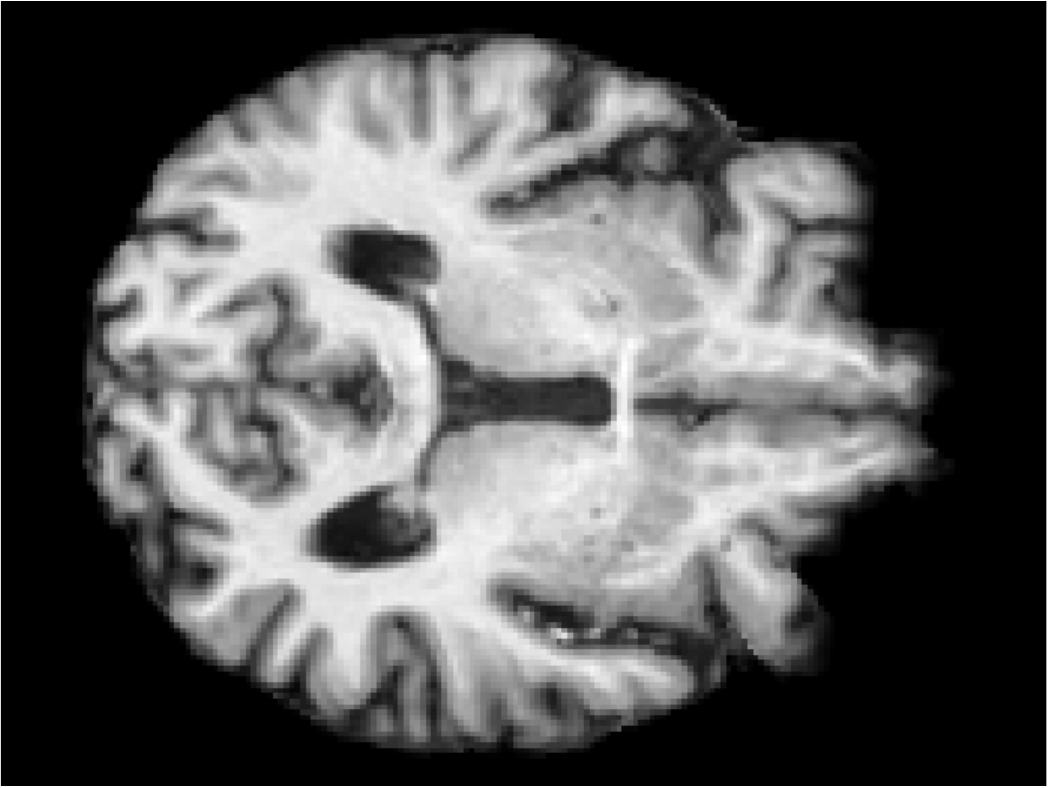} &
        \mrbrains{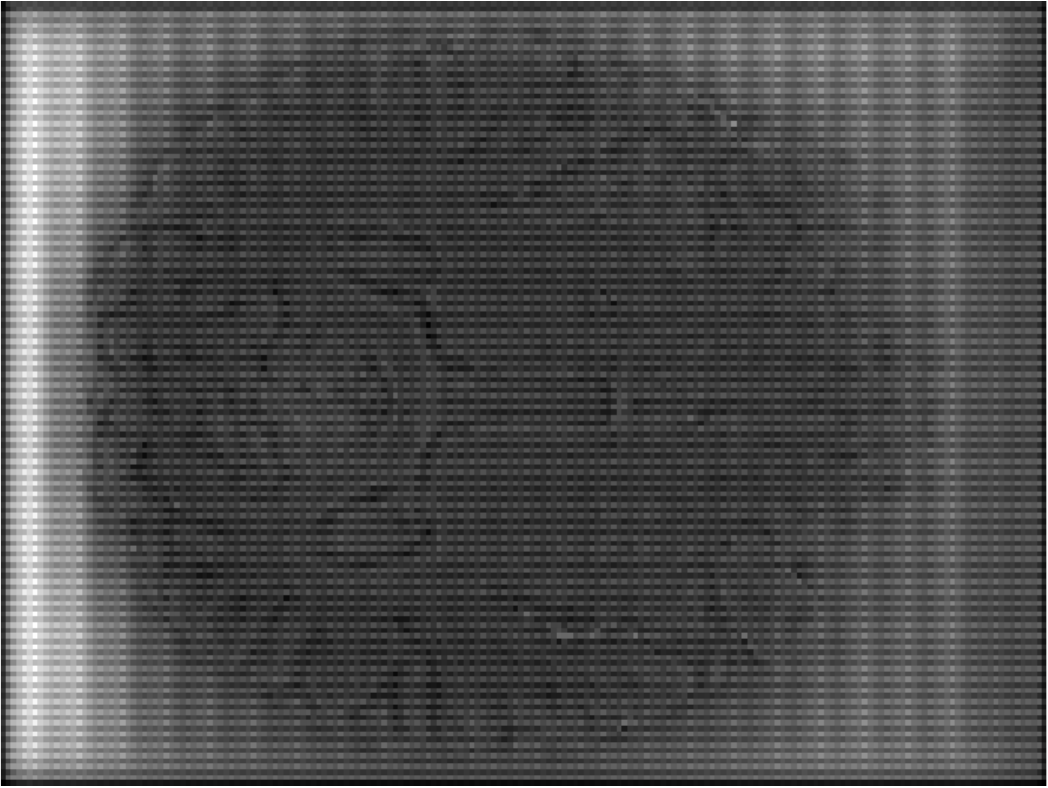} &
        \mrbrains{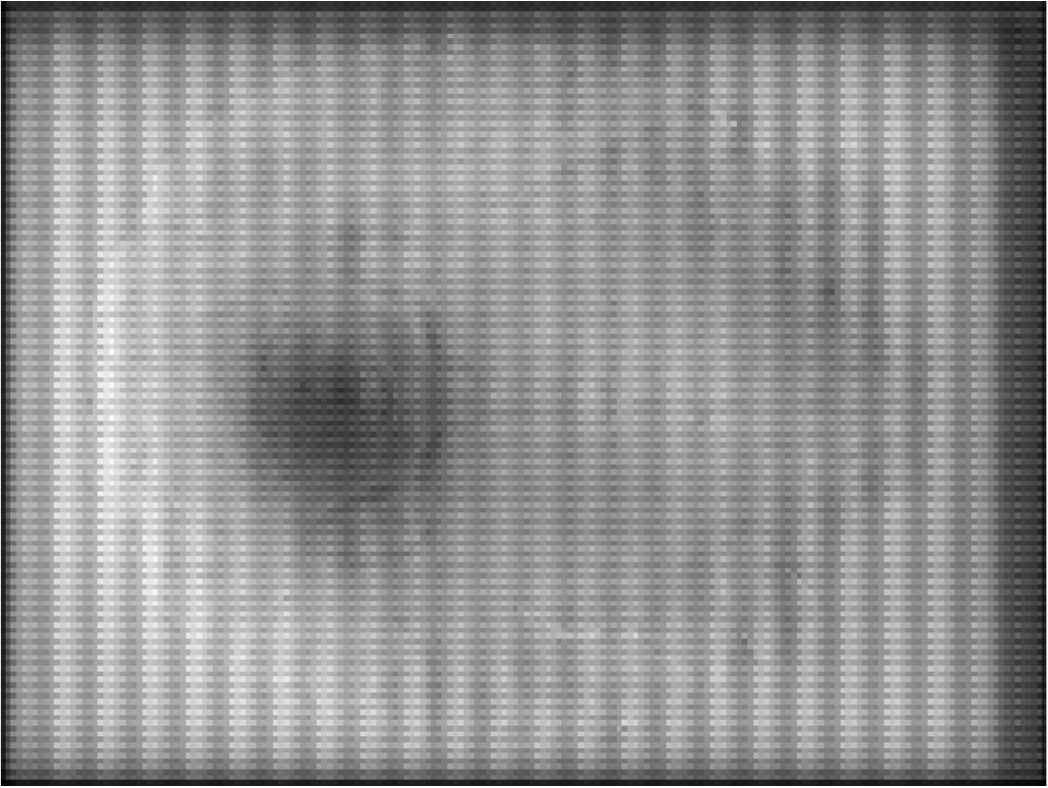} &
        \mrbrains{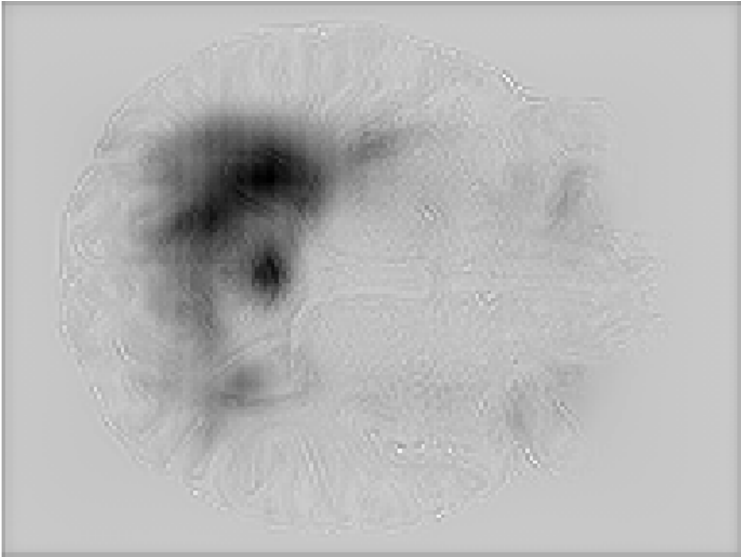} &
        \mrbrains{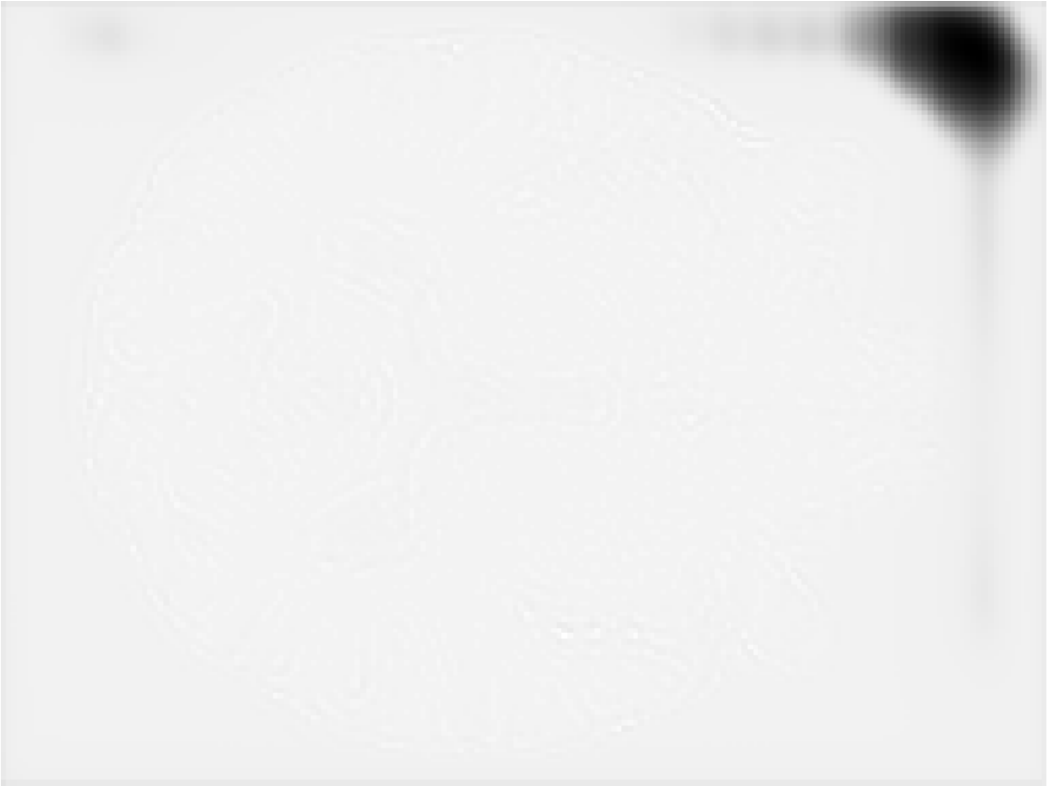} \\
        \mrbrains{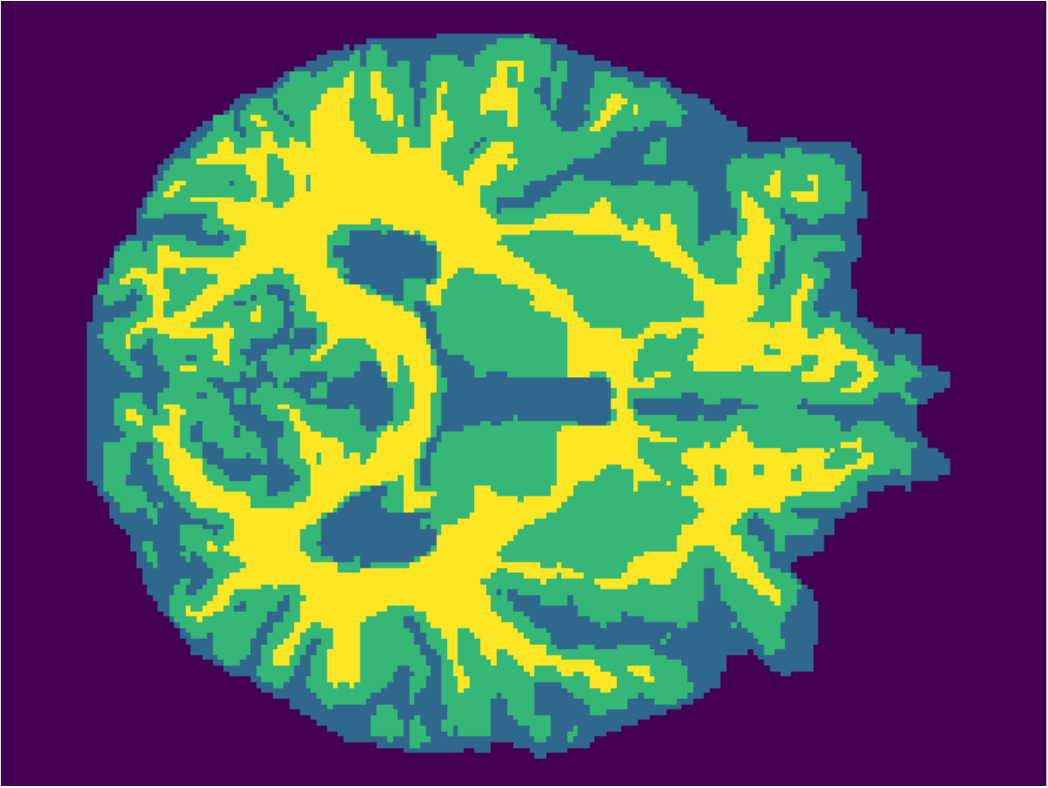} &
        \mrbrainsretrain{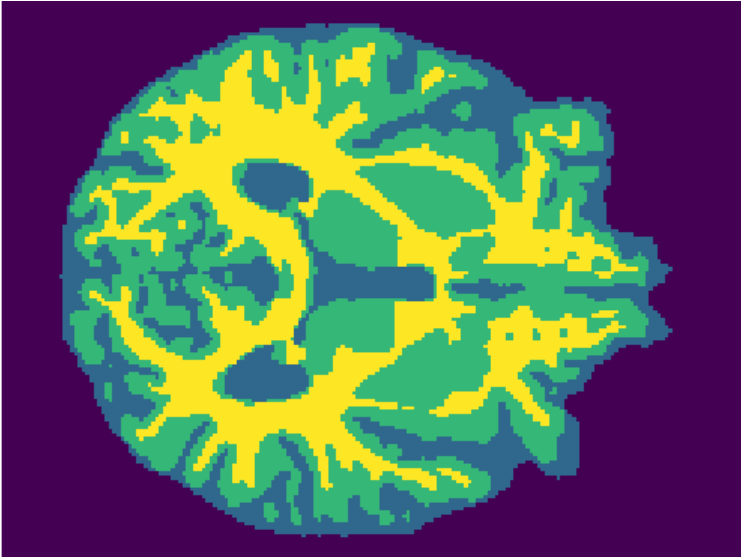} &
        \mrbrainsretrain{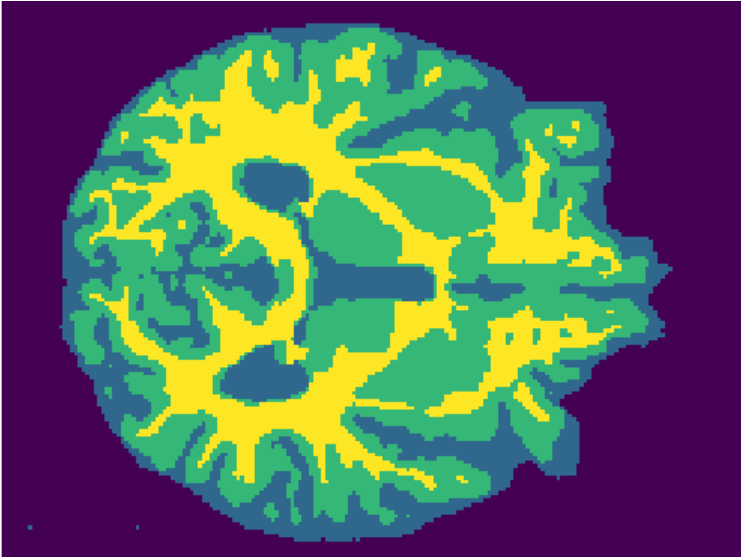} &
        \mrbrainsretrain{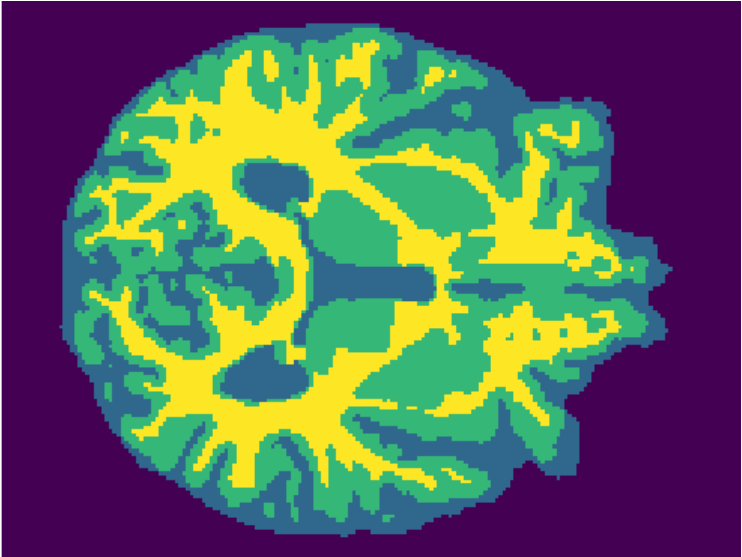} &
        \mrbrainsretrain{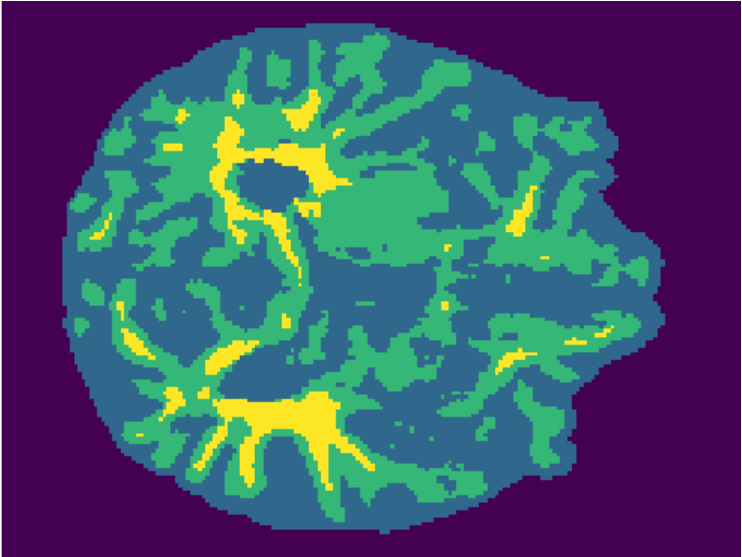} \\
        Non-encoded\,/\,GT & $\lambda$\,=\,1 & $\lambda$\,=3 & $\lambda$\,=\,10 & $\lambda$\,=\,100
    \end{tabular}    
    \end{small}
    \caption{Test results on the MRBrainS dataset. [\textbf{First column}] (Top row): input MRI image $\img$, (Bottom row): ground truth segmentation map $\gt$. [\textbf{Remaining columns}], (Top row): encoded image $G(\img)$, (Bottom row): predicted segmentation $\seg$ with re-trained segmentation networks on MRBrainS}
    \label{fig:mrbrains}
\end{figure*}

\subsubsection{Adversarial results}
\label{sec:adversarial-res}

We next report results of our adversarial approach obtained with different values of parameter $\lambda$, which controls the trade-off between segmentation accuracy and identity obfuscation. The first row of Fig.~\ref{fig:impact-lambda} shows encoded images $G(\img)$ with the corresponding raw input MRI $\img$. As can be seen, the larger the $\lambda$ value is, the more distorted the encoded image gets. Nonetheless, except for extreme cases (e.g., $\lambda$\,=\,100) the encoded images contain enough information for the segmentation network to recover a good segmentation map (c.f., the second row of Fig.~\ref{fig:impact-lambda}). The obfuscating power of our method is also illustrated by the MS-SSIM plots (c.f., third row of Fig.~\ref{fig:impact-lambda}). As $\lambda$ increases, the distribution of inter-subject MS-SSIM between encoded images $G(\img_i)$ and $G(\img_j)$ becomes more and more similar to that of intra-subjects.  

The encoder's ability to obfuscate identity is evaluated quantitatively in Table~\ref{table:ClassificationResults}. Four different techniques are used to measure this property. First, based on the observation that the distribution of MS-SSIM values differs between images from the same patient and images from different patients (c.f., last row of Fig.~\ref{fig:impact-lambda}), we compute the accuracy obtained by the best possible tresholding of MS-SSIM values (i.e., values below or equal to the threshold are predicted as same-patient images, and those above as different-patient images). Second, we report the classification accuracy of the discriminator used for training the encoder, denoted as $D_{\mr{adv}}$ in Table~\ref{table:ClassificationResults}. Third, since the encoder was trained to fool $D_{\mr{adv}}$, we also trained a new ResNet discriminator ($D_{\mr{new}}$) as in \cite{DBLP:journals/corr/HeZRS15} on the fixed encoded images to measure how good the encoder is with respect to an independent network that was not involved in training our system. Last, to assess whether an encoded image can be used to find the corresponding subject with a retrieval approach, we considered the embedding of Siamese discriminator $D_{\mr{new}}$ as representation of each encoded test image and used Euclidean distance to find most similar encoded images. We employ mAP to measure retrieval performance.

Results in Table~\ref{table:ClassificationResults} show the same trend for all four obfuscation measures. When images are not encoded, identifying the subject's identity either by comparing two images or using a retrieval-based approach is fairly easy. However, this becomes much harder for encoded images, with accuracy and mAP rates dropping as $\lambda$ increases. Moreover, as shown in column $D_{\mr{new}}$, employing a discriminator trained independently from the encoder does not help re-identify the subject's ID. This demonstrates the robustness of our method to classification approaches.

\update{ 
The segmentation performance obtained with different privacy-segmentation trade-off, defined by the $\lambda$ parameter, is given in Table~\ref{table:SegmentationResults}. As expected, the Dice score degrades when increasing $\lambda$ values, since a greater importance is then given to identity obfuscation compared to segmentation. Nevertheless, the segmentation performance on encoded images is still sufficient for many medical applications, especially when using $\lambda$\,=\,1 or $\lambda$\,=\,3. Although the definition of suitable performance is application-dependent, some authors have reported DSC values of 70\% \cite{363096,article1,pmid23336255,article2} or 80\% \cite{Mattiucci2013AutomaticDF} as threshold for clinically-acceptable segmentations. For $\lambda$\,=\,1, the overall difference in DSC compared to segmentation of non-encoded images is only 3.6\%, an impressive result considering that subject identity information is largely removed in encoded images for this $\lambda$ value.

Anatomical information capturing the shape of brain regions and cortical folds (i.e., sulci) can be predictive of both subject identity and segmentation contours. This can be observed in Fig.~\ref{fig:impact-lambda} (first row), where small anatomical details are visible in the encoded images, especially for $\lambda$\,=\,1. To obfuscate identity, the encoder must therefore produce strong noise and artifacts that dominate this morphological information. The impact of this noise in encoded images can be seen in the last row of Fig.~\ref{fig:impact-lambda}, in which the histogram of MS-SSIM scores between different-subject images (i.e., inter-subject) is pushed towards the one for same-subject images (i.e., intra-subject).
} 


\subsubsection{Generalization to new dataset}

\update{
In previous experiments, we considered the scenario where the encoder and segmentation network are trained once with some available data, and then clients send encoded images to the server for segmentation. However, this approach may fail when trying to process images different from those seen in training, for instance, coming from another hospital or acquired with different parameters. In this section, we show that \revision{the privacy-preserving encoder learned with a given dataset can be used to encode images from another dataset, and that these encoded images are useful to update the segmentation network.} 

To test this configuration, we consider the same encoder as before, which was trained using the longitudinal data from the PPMI dataset, and use it as an identity obfuscation module for clients with other data. To simulate this other data source, we used images from the MRBrainS dataset which were acquired with a different acquisition protocol than PPMI and have three labels instead of five, i.e., WM, GM, and CSF.} 
We first tested on MRBrainS images our model pre-trained with PPMI data. In order to match the three-class ground-truth, we merged the CSF int. and CSF ext. outputs into a single CSF class, and the GM and nuclei outputs into a single GM class. Results for different $\lambda$ values are shown at the top of Table~\ref{table:MRBrainsSegmentationResults}. As expected, these results are slightly lower than those on PPMI (see Table~\ref{table:SegmentationResults}).
\revision{Interestingly, we observe a small improvement in the overall Dice when using encoded images, compared to the baseline network trained with non-encoded PPMI images. This suggests that the system trained with adversarial loss generalizes better to new data than the baseline segmentation network, despite the heavy distortion of encoded images. This improved generalization of our method is due to optimizing the encoder for both segmentation accuracy and identity obfuscation, which possibly removes site-related variability (e.g., intensity distribution, tissue contrast, etc.) from encoded images.
}

As can be seen at the bottom of Table~\ref{table:MRBrainsSegmentationResults}, segmentation accuracy improves when the segmentation network is retrained on MRBrainS data following a distributed learning schedule. This shows that the segmentation network of our system can be updated, even after being deployed onto a cloud server. Segmentation maps as well as encoded MRBrainS images are given in Fig~\ref{fig:mrbrains}. \revision{However, like other deep learning methods, when the discrepancy between domains becomes too large (i.e., different image modalities), the system would need to be retrained end-to-end.}



\begin{table}[tp]
\centering
\caption{Segmentation Dice score on the MRBrainS dataset for different values of $\lambda$. (Top) the CNNs have not been retrained while (Bottom) the segmentation network has been retrained following a distributed learning approach.}
\label{table:MRBrainsSegmentationResults}
\small
\begin{tabular}{ll|ccc|c}
    \toprule
    & &  GM & WM & CSF & Overall\\    
    \midrule
    \multirow{5}{*}{No retrain}   
    & Non-enc & 0.742 & 0.805 & 0.778 & 0.783\\
    & $\lambda$\,=\,1 & 0.768 & 0.822 & 0.804 & 0.796 \\
    & $\lambda$\,=\,3 & 0.767 & 0.852 & 0.798 & 0.804 \\
    & $\lambda$\,=\,10 & 0.757 & 0.798 & 0.768 & 0.772 \\
    & $\lambda$\,=\,100 & 0.499 & 0.464 & 0.648 & 0.537 \\
    \midrule
    \multirow{5}{*}{Retrain} 
    & Non-enc & 0.832 & 0.866 & 0.840 & 0.845 \\
    & $\lambda$\,=\,1 & 0.819 & 0.827 & 0.823 & 0.821 \\
    & $\lambda$\,=\,3 & 0.794 & 0.807 & 0.831 & 0.814 \\
    & $\lambda$\,=\,10 & 0.780 & 0.747 & 0.797 & 0.790 \\
    & $\lambda$\,=\,100 & 0.605 & 0.360 & 0.572 & 0.586 \\
    \bottomrule
\end{tabular}
\end{table}

\begin{figure*}[tp]
    \centering
    \begin{small}
    \setlength{\tabcolsep}{3pt}
    \begin{tabular}{cccc}
        \robust{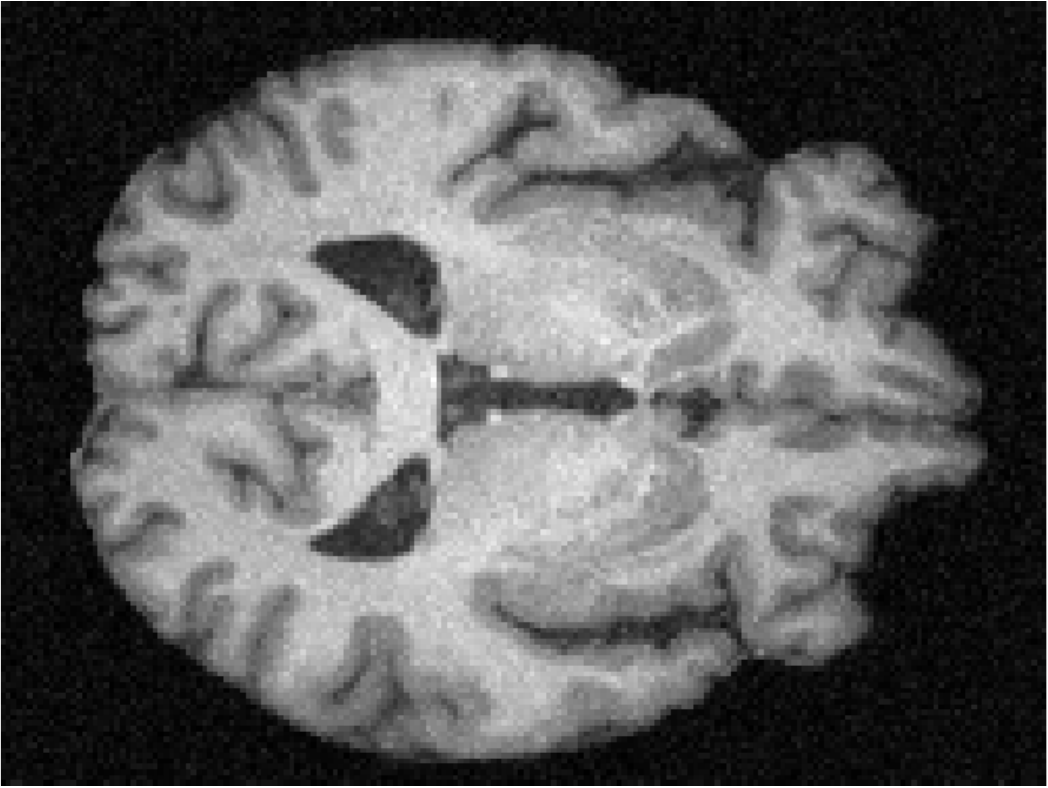} &
        \robust{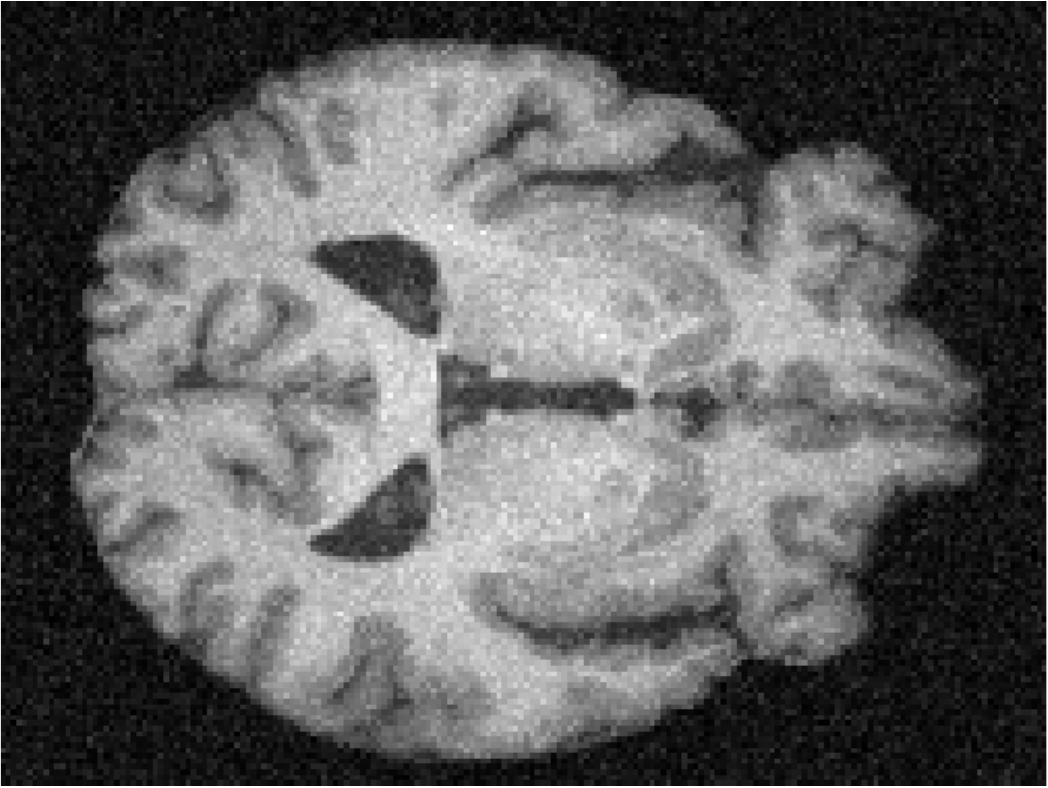} &
        \robust{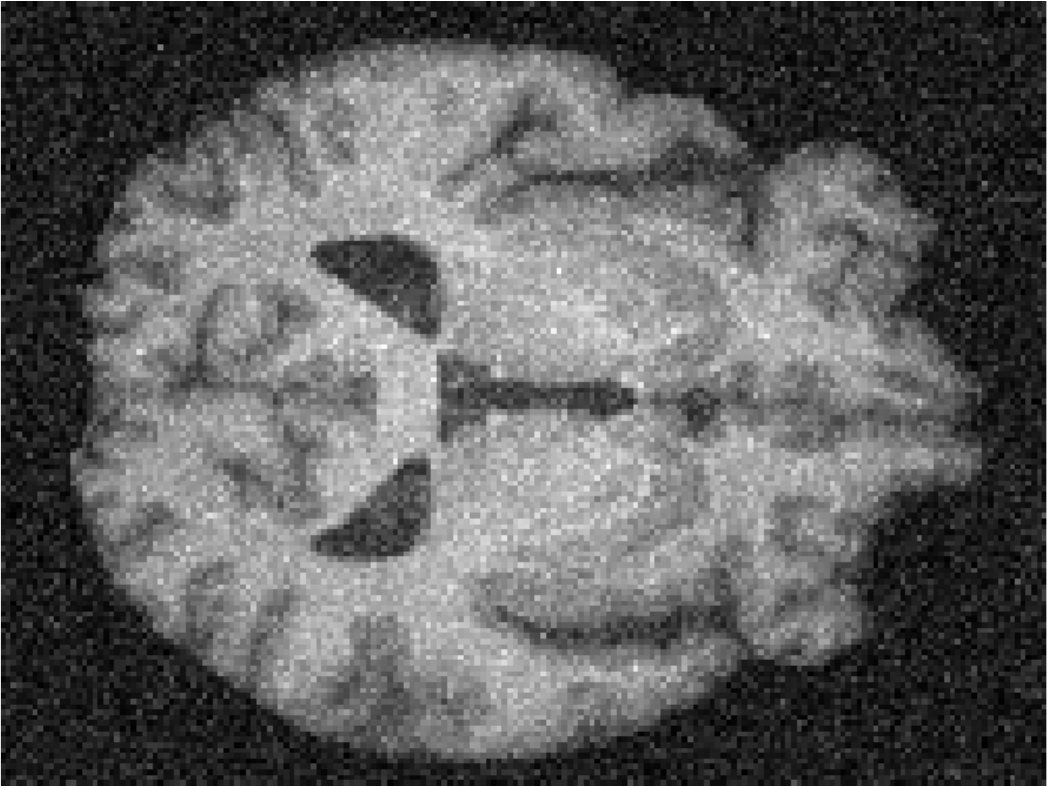} &
        \robust{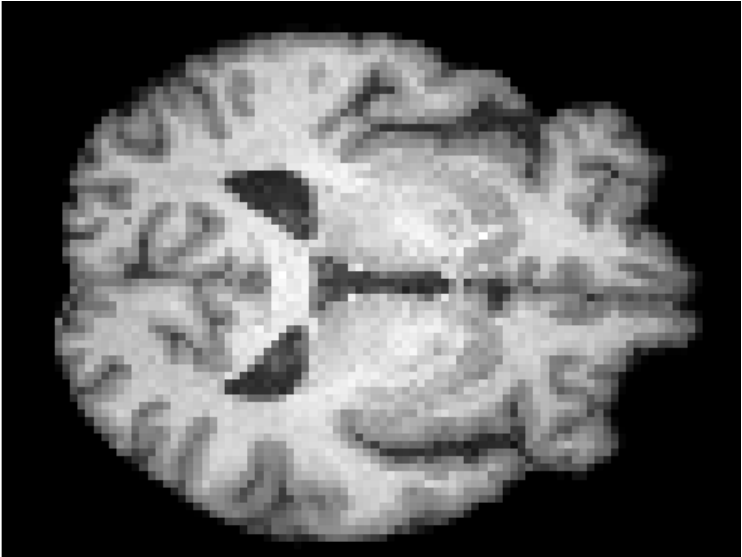} \\
        \robust{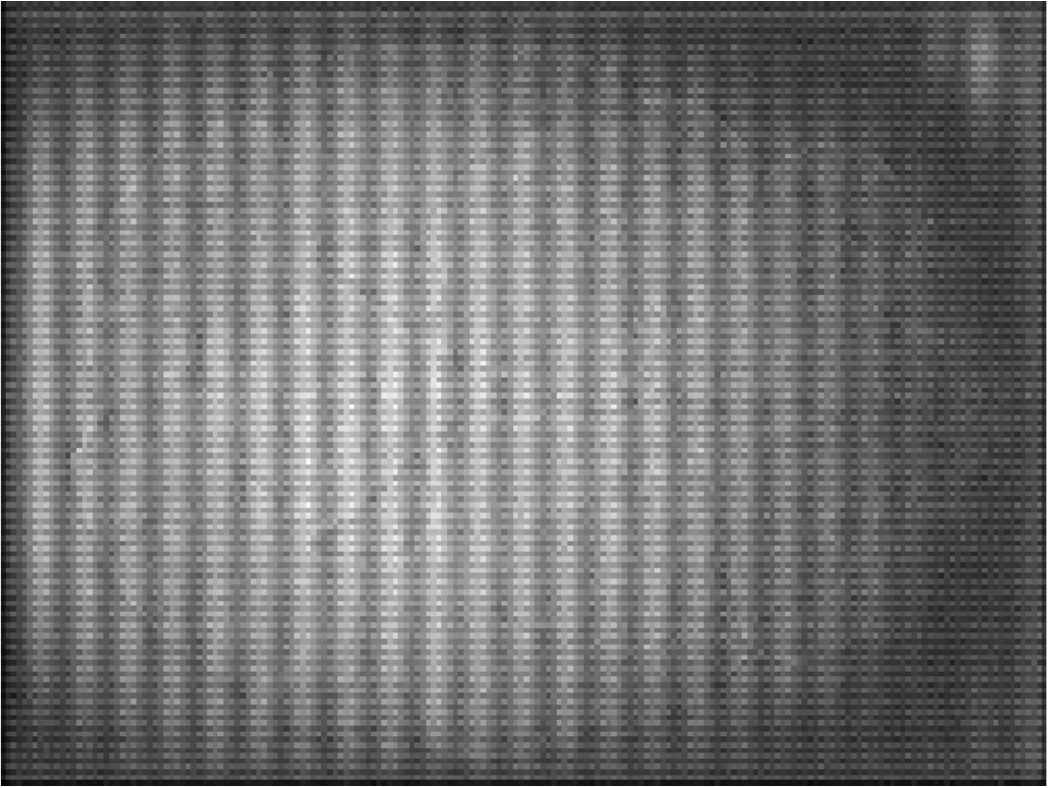} &
        \robust{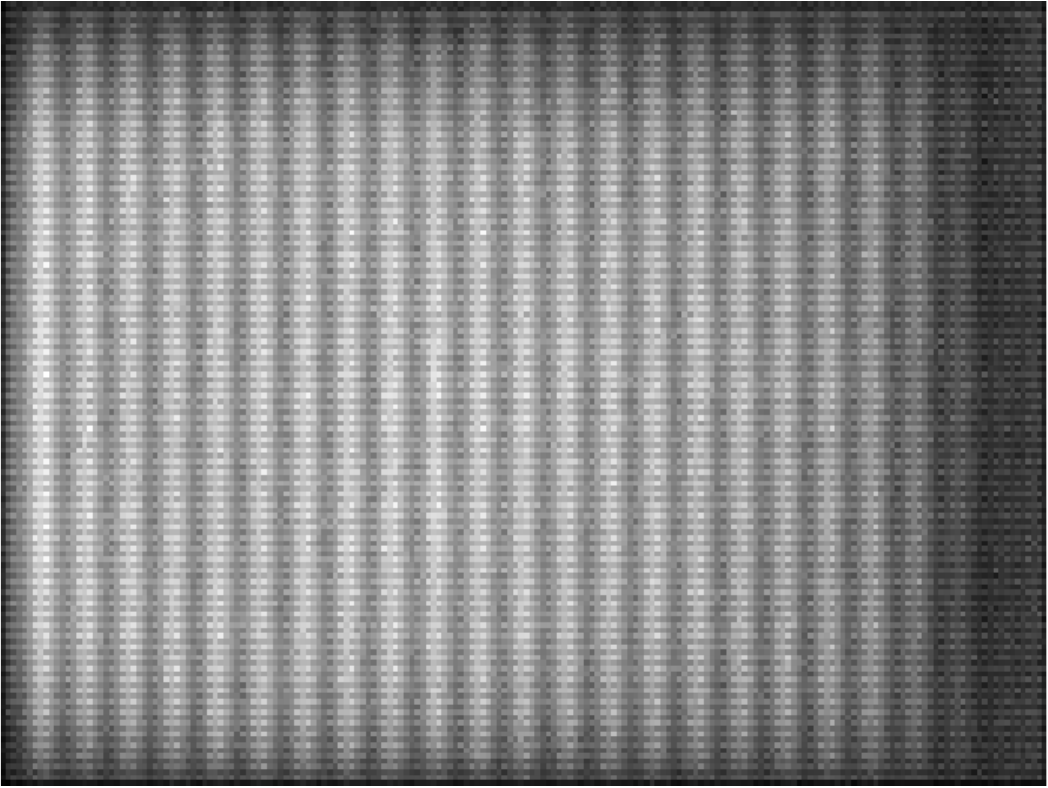} &
        \robust{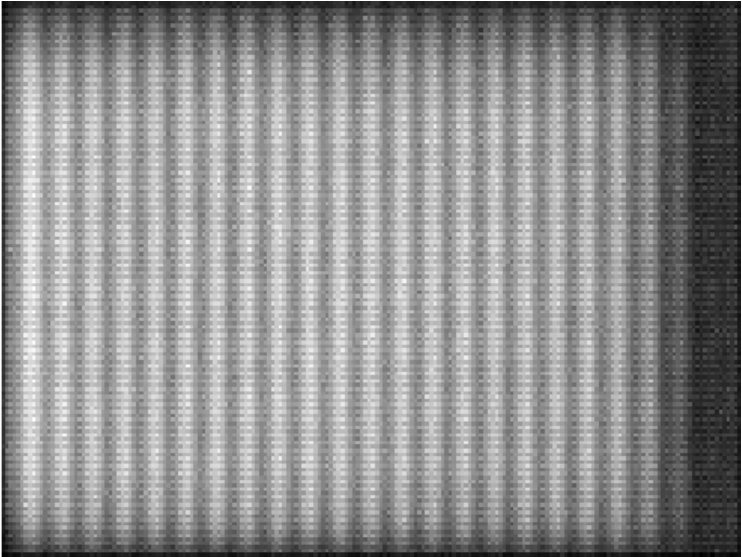} &
        \robust{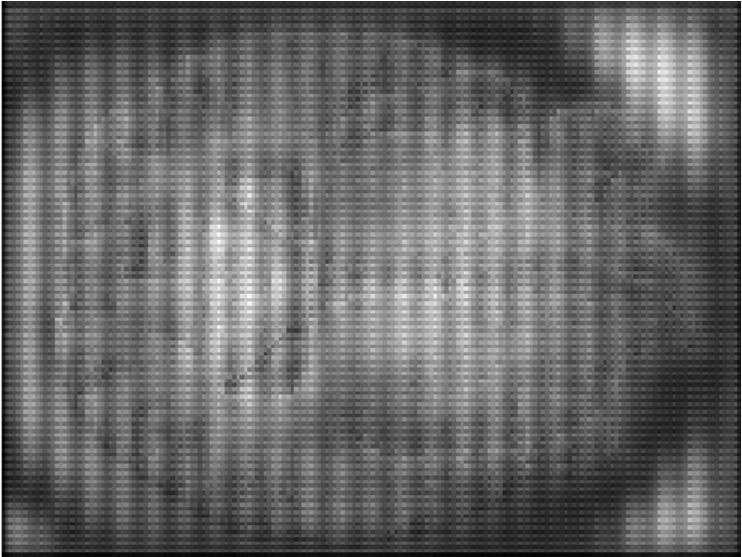} \\
        \robust{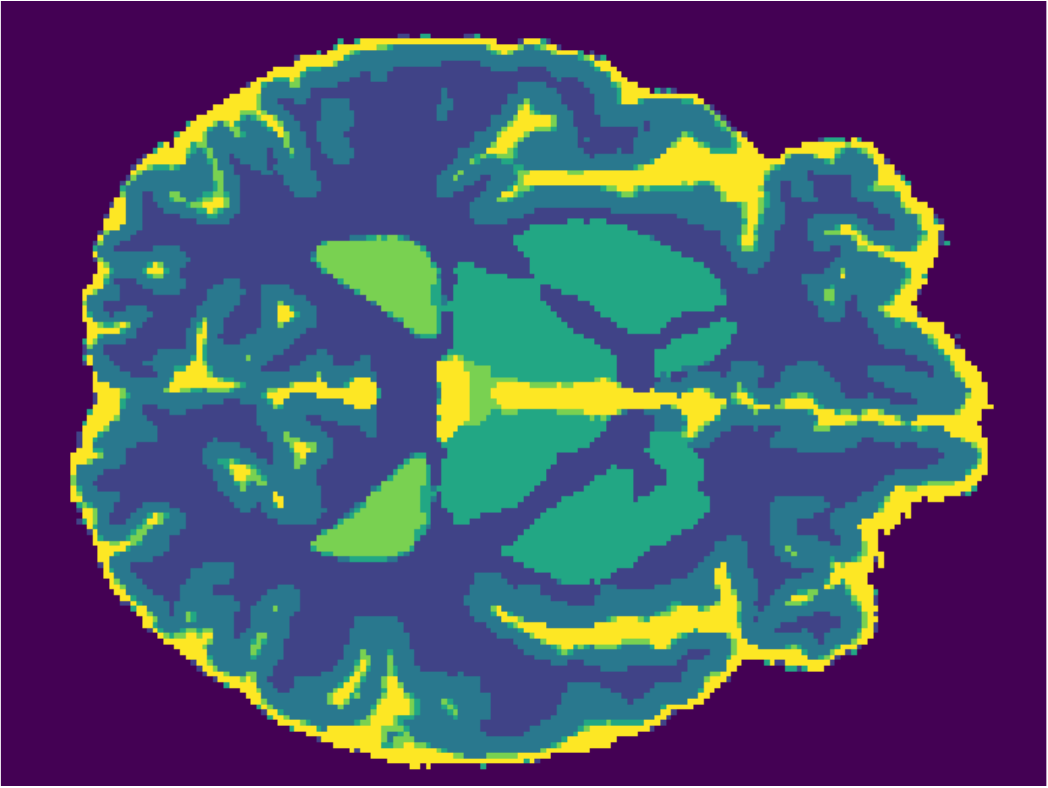} &
        \robust{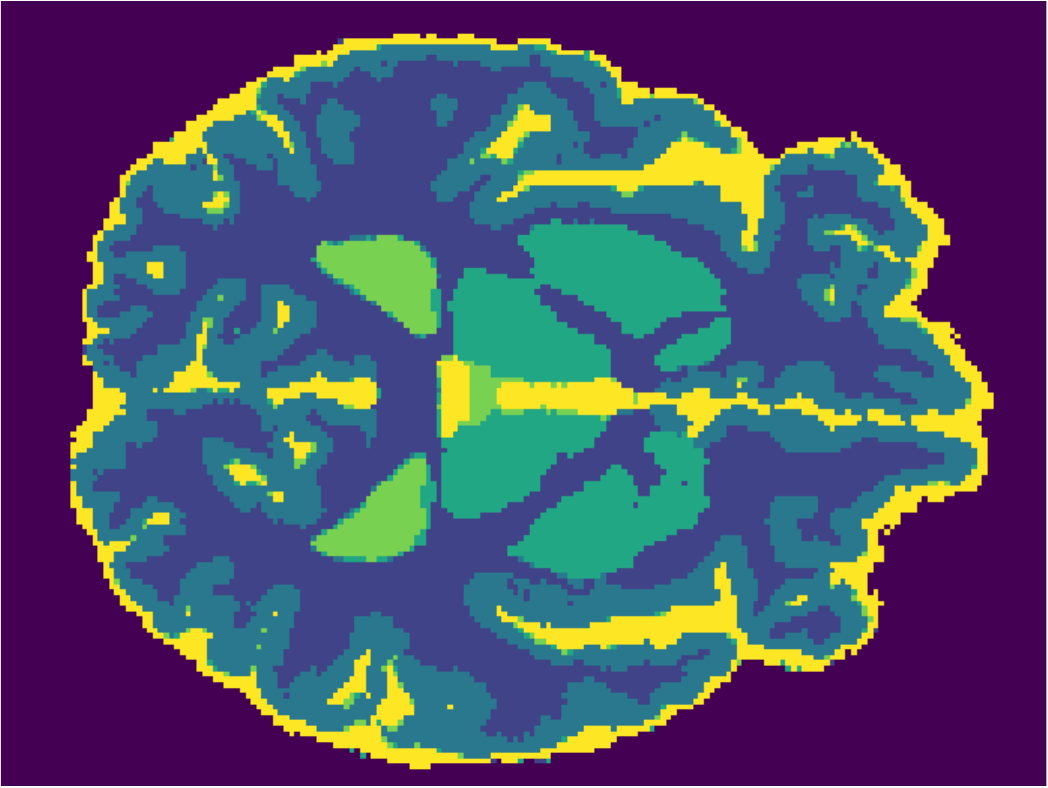} &
        \robust{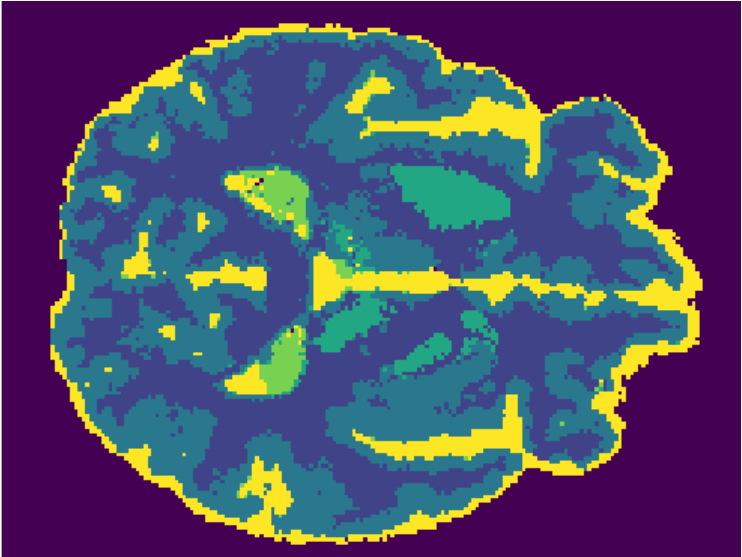} &
        \robust{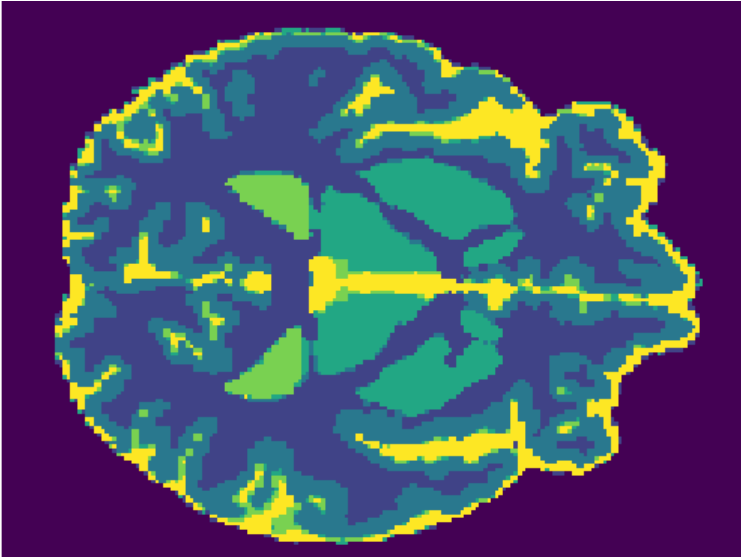} \\
        SNR = 15db & SNR = 10db & SNR = 5db & Low resolution
    \end{tabular}    
    \end{small}
    \vspace{.1cm} 
    \caption{Segmentation with different noise level and a lower resolution setting. (Top row): Degraded images, (Middle row): Encoded images, (Bottom row): Segmentation Results}
    \label{fig:Robustness}
\end{figure*}

\subsubsection{Robustness analysis}

To make sure that our system does not work only on high-quality images such as those of PPMI, we performed a robustness analysis where we trained our method on the original PPMI dataset (with $\lambda$\,=\,1) and tested it on noisy versions of the PPMI test images or on images with a lower 2\ttimes2\ttimes2 cm$^3$ resolution. Results are provided in Fig.~\ref{fig:Robustness} and Table~\ref{table:Robustness}. As can be seen, the encoding and segmentation is not much affected by noise. The segmentation accuracy is close to the one obtained on the original PPMI test set, with overall Dice scores around 80\% for noisy images with an SNR of 15db or 10db. Reducing image resolution seems to induce more significant changes in the encoding, however the segmentation network appears robust to these changes.

\update{The robustness of our model to noise and resolution can be explained as follows. Due to the adversarial optimization between the encoder and discriminator, the segmentation network sees a wide variety of distortion patterns in encoded images as the encoder tries to fool the discriminator. By forcing the segmentation network to produce the same output for these different distorted inputs, we regularize training and make this network more robust to noise and resolution. This principle is at the core of powerful regularization techniques for semi-supervised learning, such as Virtual Adversarial Training (VAT)~\cite{miyato2018virtual}. 
}

\begin{table}[tp]
\centering
\caption{Segmentation Dice score on PPMI dataset with different levels of Rician noise (measured in dB) and low resolution setting.}
\label{table:Robustness}
\small\setlength{\tabcolsep}{3pt}
\begin{tabular}{ l|ccccc|c}
    \toprule
    &  GM & WM & Nuclei & CSF int. & CSF ext. & Overall \\
    \midrule
    Noise (15\,dB)~ & 0.921 & 0.818 & 0.572 & 0.585 & 0.743 & 0.808 \\
    Noise (10\,dB)~ & 0.917 & 0.804 & 0.566 & 0.582 & 0.696 & 0.797 \\
    Noise (5\,dB)~ & 0.821 & 0.706 & 0.514 & 0.472 & 0.331 & 0.669 \\
    \midrule
    Low res. & 0.881 & 0.781 & 0.552 & 0.516 & 0.683 & 0.759 \\
    \bottomrule
\end{tabular}
\end{table}

\subsubsection{Dimension of encoded images}

\textcolor{black}{By default, our encoder outputs a one-channel feature map (c.f., image $G(x_i)$ in Fig.~\ref{fig:system}). Our motivation for using a single channel is that the encoder should preserve the amount of information while transforming the input. 
To validate this hypothesis, we repeated the same experiment but with a larger number of channels for the encoded images. The intra-subject and inter-subject prediction results are reported in Table~\ref{table:multichannels-classification} and the segmentation Dice scores are in Table~\ref{table:multichannels-segmentation}. 
These results show that, even though the segmentation performance slightly increases when encoding images with multiple channels, privacy preservation is compromised, as shown by the significant increase in MS-SSIM scores, discriminator accuracy ($D_{\mr{adv}}$) and mAP.}

\begin{table}[tp]
\centering 
\caption{Intra\update{-subject} and inter-subject prediction accuracy on test examples using different numbers of channels in encoded images.}
\label{table:multichannels-classification}
\small\setlength{\tabcolsep}{5pt}
\begin{tabular}{ c|cc|c}
    \toprule
    Num. of channels &  MS-SSIM & $D_{\mr{adv}}$ & mAP \\
    \midrule
    1 & 0.564 & 0.520 & 0.189 \\
    2 & 0.658 & 0.557 & 0.230 \\
    3 & 0.642 & 0.541 & 0.216 \\
    \bottomrule
\end{tabular}
\end{table}

\subsubsection{Advantage of a Siamese discriminator}

\color{black}To assess the benefit of using a Siamese discriminator in our model, we replaced it by a multi-class classifier with the same pre-trained DenseNet CNN backbone as the Siamese network. In this new model, each subject is given a different class ID and the multi-class discriminator has to predict the class ID of encoded images. Unlike for the Siamese network, training this classifier requires to have images from the same subject in both the training set tand validation set. Hence, we divided the original dataset into a new training set and validation set, and used this new split to retrain the whole system. Multi-class cross-entropy was used as loss function for the discriminator. Since the number of classes is not known in advance (i.e., new subjects can be added to the system after training) and the number of samples per class is very limited (1--4 images per subject), the classifier's output cannot be used directly to evaluate its ability to identify new subjects in testing. Instead, we considered the same out-of-sample strategy as with the Siamese discriminator, and used the features obtained from the last convolutional layer as an embedding for a nearest-neighbor retrieval analysis.


Results in Table~\ref{table:multi-classes} show that the multi-class classifier is worse than the Siamese network at recovering identity in non-encoded images, i.e., 0.442 vs 0.850 in terms of mAP. However, when used to train the entire network, the patient identity is more easily recovered when the images are encoded with the multi-class classifier (0.360) than by the Siamese network (0.189). This motivates the use of  a Siamese network as discriminator in our model.

\color{black}


\subsubsection{Top-k retrieval analysis on segmentation maps}\label{sec:retrieval-seg}

\begin{table}[tp]
\centering
\caption{\textcolor{black}{Top-k retrieval analysis results (mAP) using segmentation maps}}
\label{table:seg_id}
\small\setlength{\tabcolsep}{3pt}
{\color{black}
\begin{tabular}{l|c|c|c|c}
    \toprule
    & $\lambda$\,=\,1 & $\lambda$\,=\,3 & $\lambda$\,=\,10 & $\lambda$\,=\,100 \\
    \midrule
    Full-size segmentation maps &
    0.826 & 0.811 & 0.804 & 0.592\\
    64\ttimes64\ttimes64 patches & 0.716 & 0.683 & 0.697 & 0.398 \\
    32\ttimes32\ttimes32 patches & 0.632 & 0.624 & 0.583 & 0.240  \\
    \bottomrule
\end{tabular}
}
\end{table}

\update{

The goal of our privacy-preserving method is to segment medical images on a server without having to provide sensitive information about the subject. For example, sending non-encoded images could reveal the gender and age of a subject, or if this subject suffers from a neurological disease/disorder like Alzheimer's. Our method achieves this by distorting the input image with noise patterns so that 1) visual interpretation is nearly impossible and 2) identity cannot be recovered easily. However, because we encode the input, but not the output segmentation, an attacker still has access to some information that can help determine the subject's identity and condition.     

In the next experiment, we evaluate whether subject identity can be recovered from the segmentation network's output, using a top-k retrieval analysis similar to the one presented in Section \ref{sec:adversarial-res}. For this analysis, we suppose that an attacker compares the segmentation map of a test image against those of training images, and identifies the subject as the one corresponding to the most similar image. Since a direct pixel-to-pixel comparison is highly sensitive to transformations such as translation, scaling and rotation, as in the previous retrieval analyses, we instead use the representation of a Siamese discriminator as feature vector for matching.      

Results of are reported in Table \ref{table:seg_id}. The first row gives the mAP when using whole-image segmentation maps, for different values of $\lambda$. Parameter $\lambda$ affects the retrieval indirectly since a higher value leads to a less accurate segmentation and, thus, a noisy representation for matching. For a small $\lambda$\,=\,1, we get an mAP of 0.826 similar to the one of 0.850 obtained for non-encoded images (c.f., Table~\ref{table:ClassificationResults}). This shows that the geometry of segmented brain structures is informative of subject identity. However, increasing $\lambda$ to the high value of 100 leads to an important drop in mAP to 0.592 caused by the poor segmentation resulting from this setting.    

As mentioned in Section \ref{sec:implementation}, an important advantage of our method is that encoded images can be cut in small 3D patches which are sent to the server in a random order for processing. This is possible because the segmentation network requires to segment an image one patch at a time. Once the client receives the segmentation output for each patch from the server, it can recover the full-size segmentation map by assembling patches following the same random order. 

We assume that reassembling the segmentation maps from randomly-permuted patches is challenging, and that potential attackers instead try to identify the subject's identity by matching patches against a database of previously seen patches. Based on this assumption, we repeat the same top-k retrieval analysis as before, except we now match the Siamese network representation of test patches with those of training patches. The second and third rows of Table~\ref{table:ClassificationResults} give retrieval mAP when employing patches of size 64\ttimes64\ttimes64 and 32\ttimes32\ttimes32, respectively. We observe that retrieval rates substantially degrade when sending encoded patches to the server, with respective mAP drops of 0.110 and 0.194 for patches of size 64\ttimes64\ttimes64 and 32\ttimes32\ttimes32, when using $\lambda$\,=\,1. This indicates that the limited spatial context of patches, compared to whole segmentation maps, renders more difficult the identification of subjects. While mAP values stay relatively similar for $\lambda$\,=\,1--10, they sharply decrease for $\lambda$\,=\,100 due to the poor segmentation obtained with this setting.

\revision{Although the non-encrypted segmentation map can be seen as a security weakness, the results in able~\ref{table:seg_id} show that the retrieval accuracy is greatly reduced when the client sends shuffled patches for segmentation. Moreover, our method not only removes patient identity from the image, but also scrubs out most of its content. Hence, even if the identity of a patient was to be recovered via the inspection of the segmentation map, all an attacker would have access to are encoded image patches and not the original non-encoded MRI.}

}




\subsubsection{Runtime analysis}

In terms of runtime, an average of 0.08 seconds is required to encode an image on a NVIDIA GTX 1080Ti, whereas the entire segmentation process requires around 0.1 seconds per 3D MRI image. This runtime is negligible compared to the 8 to 12 hours required by Freesurfer. 

\section{Discussion and conclusion}

We presented a novel framework which integrates an encoder, a segmentation CNN and a Siamese network to preserve the privacy of medical imaging data. Experimental results on two independent datasets showed that the proposed method can preserve the identity of a patient while maintaining the performance on the target task.  While this is an interesting application \textit{per se}, it opens the door to appealing potential uses. For example, this approach can be integrated in a continual learning scenario trained on a decentralized dataset, where images have to be shared across institutions but privacy needs to be preserved. From a clinical perspective, obfuscating visual data in addition to current anonymization techniques may foster multi-centre collaborations, resulting in larger datasets as well as more complete and heterogeneous clinical studies.

\begin{table}[tp]
\centering
\caption{Segmentation Dice score on encoded images with different number of channels.}
\label{table:multichannels-segmentation}
\small\setlength{\tabcolsep}{3pt}
\begin{tabular}{ c|ccccc|c}
    \toprule
    Nb channels &  GM & WM & Nuclei & CSF int. & CSF ext. & Overall \\
    \midrule
    1 & 0.925 & 0.824 & 0.580 & 0.598 & 0.752 & 0.812 \\
    2 & 0.935 & 0.851 & 0.621 & 0.633 & 0.786 & 0.824 \\
    3 & 0.932 & 0.848 & 0.617 & 0.637 & 0.802 & 0.829 \\
    \bottomrule
\end{tabular}
\end{table}

Additionally, we have shown that the proposed privacy-preserving model generalizes well to novel datasets, unlike similar works \cite{chen2018vgan} which cannot generalize to encoded images of subjects not seen in training. This facilitates the scalability of our approach to new datasets or tasks. As preliminary step towards preventing privacy leakage in medical imaging data, this study has however some limitations. For example, the domain shift between employed datasets is not significantly large, since both include MRI images of adult brains (even though the acquisition protocols and parameters across scanners differ). Although similar domain shift has resulted in a performance degradation in segmentation networks \cite{dolz20183d}, results demonstrate the good generability of the proposed method in these cases. Future investigations will explore the generalization capabilities of the trained encoder on datasets where the domain shift is larger, for example, between infant and adult brains or even between different image modalities such as MRI and CT.

\begin{table}[tp]
\centering
\caption{Image-retrieval performance (mAP) of the multi-class classifier and Siamese discriminator.}
\label{table:multi-classes}
\small\setlength{\tabcolsep}{5pt}
    \begin{tabular}{l|ccc|cc}
        \toprule
        & \multicolumn{3}{c|}{Multi-class classifier} & \multicolumn{2}{c}{Siamese discr.} \\
        \cmidrule(l{6pt}r{6pt}){2-4} \cmidrule(l{6pt}r{6pt}){5-6}
        & Train. & Valid. & Test & Train. & Test \\
        \midrule
        Non-encoded~ & 0.712 & 0.478 & 0.442 & 0.894 & 0.850 \\
        Encoded & 0.404 & 0.324 & 0.360 & 0.153 & 0.189 \\
        \bottomrule
    \end{tabular}
\end{table}

\section*{Acknowledgment}

We acknowledge the support of the Natural Sciences and Engineering Research Council of Canada (NSERC) as well as the Reseau de BioImagrie du Quebec (RBIQ) and thank NVIDIA corporation for supporting this work through their GPU grant program.

\bibliographystyle{IEEEtranE}
\bibliography{IEEEabrv,egbib}

\begin{thebibliography}{10}
\providecommand{\url}[1]{#1}
\csname url@samestyle\endcsname
\providecommand{\newblock}{\relax}
\providecommand{\bibinfo}[2]{#2}
\providecommand{\BIBentrySTDinterwordspacing}{\spaceskip=0pt\relax}
\providecommand{\BIBentryALTinterwordstretchfactor}{4}
\providecommand{\BIBentryALTinterwordspacing}{\spaceskip=\fontdimen2\font plus
\BIBentryALTinterwordstretchfactor\fontdimen3\font minus
  \fontdimen4\font\relax}
\providecommand{\BIBforeignlanguage}[2]{{%
\expandafter\ifx\csname l@#1\endcsname\relax
\typeout{** WARNING: IEEEtran.bst: No hyphenation pattern has been}%
\typeout{** loaded for the language `#1'. Using the pattern for}%
\typeout{** the default language instead.}%
\else
\language=\csname l@#1\endcsname
\fi
#2}}
\providecommand{\BIBdecl}{\relax}
\BIBdecl

\bibitem{Zhou2017}
S.~Zhou, H.~Greenspan, and D.~Shen, \emph{Deep Learning for Medical Image
  Analysis}.\hskip 1em plus 0.5em minus 0.4em\relax Elsevier Science, 2017.

\bibitem{litjens2017survey}
G.~Litjens \emph{et~al.}, ``A survey on deep learning in medical image
  analysis,'' \emph{MedIA}, vol.~42, pp. 60--88, 2017.

\bibitem{Kumar2018}
K.~Kumar, M.~Toews, L.~Chauvin, O.~Colliot, and C.~Desrosiers, ``Multi-modal
  brain fingerprinting: A manifold approximation based framework,''
  \emph{NeuroImage}, vol. 183, pp. 212 -- 226, 2018.

\bibitem{McMahan2017}
H.~McMahan, E.~Moore, D.~Ramage, S.~Hampson, and B.A., ``Fast trust region for
  segmentation,'' in \emph{Proc of ICAIS}, 2017, pp. 1273--1282.

\bibitem{Dowlin16}
N.~Dowlin, R.~Gilad-Bachrach, K.~Laine, K.~Lauter, M.~Naehrig, and J.~Wernsing,
  ``A reproducible evaluation of {ANT}s similarity metric performance in brain
  image registration.'' in \emph{proc of ICML}, 2016.

\bibitem{Hesamifard17}
E.~Hesamifard, H.~Takabi, and M.~Ghasemi, ``Crypto{DL}: Deep neural networks
  over encrypted data,'' \emph{arXiv preprint arXiv:1711.05189}, 2017.

\bibitem{nandakumar2019towards}
K.~Nandakumar, N.~Ratha, S.~Pankanti, and S.~Halevi, ``Towards deep neural
  network training on encrypted data,'' in \emph{proc of CVPR-W}, 2019, pp.
  0--0.

\bibitem{Hardy2017}
S.~Hardy \emph{et~al.}, ``Private federated learning on vertically partitioned
  data via entity resolution and additively homomorphic encryption,''
  \emph{ArXiv}, vol. abs/1711.10677, 2017.

\bibitem{DBLP:journals/corr/RonnebergerFB15}
O.~Ronneberger, P.~Fischer, and T.~Brox, ``U-net: Convolutional networks for
  biomedical image segmentation,'' in \emph{International Conference on Medical
  image computing and computer-assisted intervention}, 2015, pp. 234--241.

\bibitem{rouhani2018}
B.~Rouhani, S.~Riazi, and F.~Koushanfar, ``Deep{S}ecure: Scalable
  provably-secure deep learning,'' in \emph{in proc of Design Auto. Conf.
  (DAC)}, 2018.

\bibitem{DBLP:journals/corr/LucCCV16}
P.~Luc, C.~Couprie, S.~Chintala, and J.~Verbeek, ``Semantic segmentation using
  adversarial networks,'' in \emph{NIPS Workshop on Adversarial Training},
  2016.

\bibitem{goodfellow2014generative}
I.~Goodfellow \emph{et~al.}, ``Generative adversarial nets,'' in \emph{proc
  NIPS}, 2014, pp. 2672--2680.

\bibitem{ganin2016domain}
Y.~Ganin \emph{et~al.}, ``Domain-adversarial training of neural networks,''
  \emph{JMLR}, vol.~17, no.~1, pp. 2096--2030, 2016.

\bibitem{oleszkiewicz2018siamese}
W.~Oleszkiewicz, P.~Kairouz, K.~Piczak, R.~Rajagopal, and T.~Trzci{\'n}ski,
  ``Siamese generative adversarial privatizer for biometric data,'' in
  \emph{proc of ACCV}, 2018, pp. 482--497.

\bibitem{ziad2016cryptoimg}
M.~T.~I. Ziad, A.~Alanwar, M.~Alzantot, and M.~Srivastava, ``Crypto{I}mg:
  Privacy preserving processing over encrypted images,'' in \emph{proc of IEEE
  CNS}, 2016, pp. 570--575.

\bibitem{wang2017encrypted}
W.~Wang, C.-M. Vong, Y.~Yang, and P.-K. Wong, ``Encrypted image classification
  based on multilayer extreme learning machine,'' \emph{MSSP}, vol.~28, no.~3,
  pp. 851--865, 2017.

\bibitem{paillier1999public}
P.~Paillier, ``Public-key cryptosystems based on composite degree residuosity
  classes,'' in \emph{proc ICTACT}, 1999, pp. 223--238.

\bibitem{hsu2011homomorphic}
C.-Y. Hsu, C.-S. Lu, and S.-C. Pei, ``Homomorphic encryption-based secure
  {SIFT} for privacy-preserving feature extraction,'' in \emph{proc of
  MWSF-III}, vol. 7880, 2011.

\bibitem{dai2015towards}
J.~Dai, B.~Saghafi, J.~Wu, J.~Konrad, and P.~Ishwar, ``Towards
  privacy-preserving recognition of human activities,'' in \emph{proc of ICIP},
  2015, pp. 4238--4242.

\bibitem{chen2016estimating}
J.~Chen, J.~Wu, K.~Richter, J.~Konrad, and P.~Ishwar, ``Estimating head pose
  orientation using extremely low resolution images,'' in \emph{proc of IEEE
  SSIAI}, 2016, pp. 65--68.

\bibitem{butler2015privacy}
D.~J. Butler, J.~Huang, F.~Roesner, and M.~Cakmak, ``The privacy-utility
  tradeoff for remotely teleoperated robots,'' in \emph{Proc of ACM/IEEE
  ICHRI}, 2015, pp. 27--34.

\bibitem{jalal2012depth}
A.~Jalal, M.~Z. Uddin, and T.-S. Kim, ``Depth video-based human activity
  recognition system using translation and scaling invariant features for life
  logging at smart home,'' \emph{IEEE TCE}, vol.~58, no.~3, pp. 863--871, 2012.

\bibitem{mcclure2018distributed}
P.~McClure \emph{et~al.}, ``Distributed weight consolidation: A brain
  segmentation case study,'' in \emph{proc of NIPS}, 2018, pp. 4093--4103.

\bibitem{DBLP:journals/corr/XieBFGLN14}
P.~Xie, M.~Bilenko, T.~Finley, R.~Gilad{-}Bachrach, K.~E. Lauter, and
  M.~Naehrig, ``Crypto-nets: Neural networks over encrypted data,''
  \emph{CoRR}, vol. abs/1412.6181, 2014.

\bibitem{DBLP:journals/corr/KonecnyMRR16}
J.~Konecn{\'{y}}, H.~B. McMahan, D.~Ramage, and P.~Richt{\'{a}}rik, ``Federated
  optimization: Distributed machine learning for on-device intelligence,''
  \emph{CoRR}, vol. abs/1610.02527, 2016.

\bibitem{DBLP:journals/corr/McMahanMRA16}
H.~B. McMahan, E.~Moore, D.~Ramage, and B.~A. y~Arcas, ``Federated learning of
  deep networks using model averaging,'' \emph{CoRR}, vol. abs/1602.05629,
  2016.

\bibitem{DBLP:journals/corr/abs-1812-03288}
P.~Vepakomma, T.~Swedish, R.~Raskar, O.~Gupta, and A.~Dubey, ``No peek: {A}
  survey of private distributed deep learning,'' \emph{CoRR}, vol.
  abs/1812.03288, 2018.

\bibitem{yang2019federated}
Q.~Yang, Y.~Liu, T.~Chen, and Y.~Tong, ``Federated machine learning: Concept
  and applications,'' \emph{ACM Transactions on Intelligent Systems and
  Technology (TIST)}, vol.~10, no.~2, p.~12, 2019.

\bibitem{xu2019ganobfuscator}
C.~Xu, J.~Ren, D.~Zhang, Y.~Zhang, Z.~Qin, and K.~Ren, ``{GAN}obfuscator:
  Mitigating information leakage under gan via differential privacy,''
  \emph{IEEE TIFS}, vol.~14, no.~9, pp. 2358--2371, 2019.

\bibitem{raval2017protecting}
N.~Raval, A.~Machanavajjhala, and L.~P. Cox, ``Protecting visual secrets using
  adversarial nets,'' in \emph{proc of CVPR-W}, 2017, pp. 1329--1332.

\bibitem{pittaluga2019learning}
F.~Pittaluga, S.~Koppal, and A.~Chakrabarti, ``Learning privacy preserving
  encodings through adversarial training,'' in \emph{proc of IEEE WACV}, 2019,
  pp. 791--799.

\bibitem{wu2018towards}
Z.~Wu, Z.~Wang, Z.~Wang, and H.~Jin, ``Towards privacy-preserving visual
  recognition via adversarial training: A pilot study,'' in \emph{proc of
  ECCV}, 2018, pp. 606--624.

\bibitem{yang2018learning}
T.-Y. Yang, C.~Brinton, P.~Mittal, M.~Chiang, and A.~Lan, ``Learning
  informative and private representations via generative adversarial
  networks,'' in \emph{proc of ICBD}, 2018, pp. 1534--1543.

\bibitem{DBLP:journals/corr/abs-1904-05514}
P.~C. Roy and V.~N. Boddeti, ``Mitigating information leakage in image
  representations: A maximum entropy approach,'' in \emph{Proceedings of the
  IEEE Conference on Computer Vision and Pattern Recognition}, 2019, pp.
  2586--2594.

\bibitem{chen2018vgan}
J.~Chen, J.~Konrad, and P.~Ishwar, ``{VGAN}-based image representation learning
  for privacy-preserving facial expression recognition,'' in \emph{proc of
  CVPR-W}, 2018, pp. 1570--1579.

\bibitem{wang2018cross}
S.~Wang, Z.~Ding, and Y.~Fu, ``Cross-generation kinship verification with
  sparse discriminative metric,'' \emph{IEEE transactions on pattern analysis
  and machine intelligence}, vol.~41, no.~11, pp. 2783--2790, 2018.

\bibitem{Xia_2020_CVPR}
H.~Xia and Z.~Ding, ``Structure preserving generative cross-domain learning,''
  in \emph{Proceedings of the IEEE/CVF Conference on Computer Vision and
  Pattern Recognition (CVPR)}, June 2020.

\bibitem{dolz20183d}
J.~Dolz, C.~Desrosiers, and I.~B. Ayed, ``3{D} fully convolutional networks for
  subcortical segmentation in {MRI}: {A} large-scale study,''
  \emph{NeuroImage}, vol. 170, pp. 456--470, 2018.

\bibitem{DBLP:journals/corr/SudreLVOC17}
C.~H. Sudre, W.~Li, T.~Vercauteren, S.~Ourselin, and M.~J. Cardoso,
  ``Generalised {D}ice overlap as a deep learning loss function for highly
  unbalanced segmentations,'' in \emph{Deep learning in medical image analysis
  and multimodal learning for clinical decision support}, 2017, pp. 240--248.

\bibitem{DBLP:journals/corr/ChenDHSSA16}
X.~Chen, Y.~Duan, R.~Houthooft, J.~Schulman, I.~Sutskever, and P.~Abbeel,
  ``Info{GAN}: Interpretable representation learning by information maximizing
  generative adversarial nets,'' in \emph{Advances in neural information
  processing systems}, 2016, pp. 2172--2180.

\bibitem{kraskov2004estimating}
A.~Kraskov, H.~St{\"o}gbauer, and P.~Grassberger, ``Estimating mutual
  information,'' \emph{Physical review E}, vol.~69, no.~6, p. 066138, 2004.

\bibitem{Koch2015SiameseNN}
G.~Koch, R.~Zemel, and R.~Salakhutdinov, ``Siamese neural networks for one-shot
  image recognition,'' in \emph{proc of ICML}, 2015.

\bibitem{DBLP:journals/corr/HuangLW16a}
G.~Huang, Z.~Liu, and K.~Q. Weinberger, ``Densely connected convolutional
  networks,'' in \emph{proc CVPR}, 2017.

\bibitem{marek2011ppmi}
K.~Marek \emph{et~al.}, ``The {P}arkinson {P}rogression {M}arker {I}nitiative
  ({PPMI}),'' \emph{Progress in neurobiology}, vol.~95, no.~4, pp. 629--635,
  2011.

\bibitem{roy2019quicknat}
A.~G. Roy, S.~Conjeti, N.~Navab, C.~Wachinger \emph{et~al.}, ``Quick{NAT}: A
  fully convolutional network for quick and accurate segmentation of
  neuroanatomy,'' \emph{NeuroImage}, vol. 186, pp. 713--727, 2019.

\bibitem{mendrik2015mrbrains}
A.~M. Mendrik \emph{et~al.}, ``{MRB}rain{S} challenge: online evaluation
  framework for brain image segmentation in 3{T} {MRI} scans,'' \emph{Comp.
  Intel. and Neuro.}, vol. 2015, p.~1, 2015.

\bibitem{Avants2011}
B.~Avants, N.~Tustison, G.~Song, P.~Cook, A.~Klein, and J.~Gee, ``A
  reproducible evaluation of {ANT}s similarity metric performance in brain
  image registration.'' \emph{Neuroimage}, vol.~54, no.~3, pp. 2033--2044,
  2011.

\bibitem{Wang2003}
Z.~{Wang}, E.~P. {Simoncelli}, and A.~C. {Bovik}, ``Multiscale structural
  similarity for image quality assessment,'' in \emph{proc of IEEE ACSSC},
  2003, pp. 1398--1402.

\bibitem{Kumar209726}
K.~Kumar, M.~Toews, L.~Chauvin, O.~Colliot, and C.~Desrosiers, ``Multi-modal
  brain fingerprinting: a manifold approximation based framework,''
  \emph{NeuroImage}, vol. 183, pp. 212--226, 2018.

\bibitem{dolz2019}
J.~{Dolz}, K.~{Gopinath}, J.~{Yuan}, H.~{Lombaert}, C.~{Desrosiers}, and
  I.~{Ben Ayed}, ``{HyperDense-Net}: A hyper-densely connected {CNN} for
  multi-modal image segmentation,'' \emph{IEEE TMI}, vol.~38, no.~5, pp.
  1116--1126, 2019.

\bibitem{DBLP:journals/corr/HeZRS15}
K.~He, X.~Zhang, S.~Ren, and J.~Sun, ``Deep residual learning for image
  recognition,'' in \emph{proc of CVPR}, 2016, pp. 770--778.

\bibitem{363096}
A.~P. {Zijdenbos}, B.~M. {Dawant}, R.~A. {Margolin}, and A.~C. {Palmer},
  ``Morphometric analysis of white matter lesions in {MR} images: method and
  validation,'' \emph{IEEE TMI}, vol.~13, no.~4, pp. 716--724, 1994.

\bibitem{article1}
K.~Zou \emph{et~al.}, ``Statistical validation of image segmentation quality
  based on a spatial overlap index,'' \emph{Academic radiology}, vol.~11, pp.
  178--89, 02 2004.

\bibitem{pmid23336255}
M.~A. Gambacorta \emph{et~al.}, ``{{C}linical validation of atlas-based
  auto-segmentation of pelvic volumes and normal tissue in rectal tumors using
  auto-segmentation computed system},'' \emph{Acta Oncol}, vol.~52, no.~8, pp.
  1676--1681, Nov 2013.

\bibitem{article2}
L.~Anders, F.~Stieler, K.~Siebenlist, J.~Schäfer, F.~Lohr, and F.~Wenz,
  ``Performance of an atlas-based autosegmentation software for delineation of
  target volumes for radiotherapy of breast and anorectal cancer,''
  \emph{Journal of the European Society for Therapeutic Radiology and
  Oncology}, vol. 102, pp. 68--73, 09 2011.

\bibitem{Mattiucci2013AutomaticDF}
G.~C. Mattiucci \emph{et~al.}, ``Automatic delineation for replanning in
  nasopharynx radiotherapy: What is the agreement among experts to be
  considered as benchmark?'' \emph{Acta Oncologica}, vol.~52, pp. 1417 -- 1422,
  2013.

\bibitem{miyato2018virtual}
T.~Miyato, S.-i. Maeda, M.~Koyama, and S.~Ishii, ``Virtual adversarial
  training: a regularization method for supervised and semi-supervised
  learning,'' \emph{IEEE transactions on pattern analysis and machine
  intelligence}, vol.~41, no.~8, pp. 1979--1993, 2018.

\end{thebibliography}
\end{document}